\documentclass[10pt,final]{IEEEtran}
\usepackage{amsmath}
\usepackage{amssymb}
\usepackage{amsfonts}
\usepackage{latexsym}
\usepackage{graphicx}
\usepackage{enumerate}
\usepackage{multicol}
\usepackage{color}
\usepackage{array}
\usepackage{algorithm,algorithmic}
\usepackage{bm}
\newtheorem{lemma}{Lemma}
\IEEEoverridecommandlockouts
\allowdisplaybreaks[0]

\begin{document}
\title{Robust Design of Integrated Sensing and Communication in LEO Satellite Systems}
\author{Hezhen Yang, Xiaoming Chen, and Qi Wang
\thanks{Hezhen Yang, Xiaoming Chen, and Qi Wang are with the College of Information Science and Electronic Engineering, Zhejiang University, Hangzhou 310027, China (e-mails: \{yang\_hezhen, chen\_xiaoming, wang-qi\}@zju.edu.cn).}}
\maketitle

\begin{abstract}
With the growing demand for satellite sensing and communication, the limited wireless resources are difficult to support multiple satellite systems. Therefore, it is desired to investigate integrated sensing and communication (ISAC) in low Earth orbit (LEO) satellite systems to enable multi-functionality within a single satellite, thereby saving both spectrum and orbital resources. In this paper, a framework for ISAC in LEO satellite systems is established, where a satellite can simultaneously sense multiple targets and serve multiple communication users (CUs) over the same spectrum. Considering the limited onboard energy of satellite, a novel robust beamforming design algorithm is developed with the goal of minimizing total transmit power while satisfying the mean squared error (MSE) requirements for sensing and signal-to-interference-plus-noise ratio (SINR) requirements for communication in presence of channel phase uncertainty which exacerbates the cross-functional interference. According to theoretical analysis, the proposed algorithm for ISAC in LEO satellite systems is effective. Moreover, extensive simulations confirm the superiority of the proposed algorithm over baselines.
\end{abstract}


\section{Introduction}             
\IEEEPARstart{W}{ith} the accelerated development of numerous advanced wireless services, a huge quantity of devices across the world require access to wireless networks \cite{CRAIoT}, \cite{Internet_of}. At present, traditional terrestrial wireless networks have been widely used in cities and conventional sites, however, in remote and sparsely populated areas, they are often unable to provide efficient services due to economic and environmental constraints \cite{A_key_6G}. Specifically, these areas are limited in the construction of traditional infrastructure, which leads to information blockage and delays in decision-making. In contrast, satellite communications can overcome these limitations and achieve global coverage \cite{Research_on}. 
In this context, although geostationary Earth orbit (GEO) satellites are widely employed in communication, low Earth orbit (LEO) satellites receive more and more attentions. Owing to their lower orbital altitude, LEO satellites have relatively low propagation delays. At the same time, benefiting from the short propagation distance of LEO satellites, the signal loss caused by propagation is smaller. In addition, LEO satellites, due to their relative movement to the Earth, can be connected even if there are obstacles near the terminals \cite{LEO_Satellite}. 

Nowadays, satellite communication technology has made significant progress. In particular, the new generation LEO satellite communication system is attempting to reposition the base station function to the near-Earth orbit by utilizing terrestrial mobile communications technology, aiming to achieve efficient and direct communications between satellite platforms and ground terminals. Considering factors such as the limitation on the quantity of satellites, it is crucial to integrate satellite networks with terrestrial networks \cite{Space-terrestrial}. For example, researchers put forward an integrated terrestrial-satellite network access architecture that provides high-speed links to communication users (CUs) desiring different quality of service (QoS) through multiple physical layer technologies \cite{Ultra-Dense}. Owing to their promising prospects, LEO communication satellite systems have seen numerous commercial applications, such as SpaceX's Starlink.

In addition to communication function, since 2000, sensing technologies have been widely used in resource measurement, urban planning, agricultural development and other scenarios \cite{Progress_and_challenges}. Several studies have employed the approach of acquiring targets' reflection coefficients as a sensing method. For instance, in \cite{Optimum_power}, the authors employed the acquired reflection coefficients as distinctive features to classify different types of objects. Through the past few years, the accelerated development of technologies such as sensing satellite platform and imaging payload is changing the function of sensing satellites from observing the Earth to real-time target sensing. With low transmission latency, LEO satellites are capable of this mission. Furthermore, their reduced free-space attenuation enables more accurate sensing \cite{Analysis_of_LEO}. For instance, in \cite{GLORIA} the authors presented a continuous GEO/LEO radar image acquisition system based on a constellation consists of several GEO radar transmitting satellites and several LEO synthetic aperture radar (SAR) receiving satellites.

On the other hand, the appearance of new applications such as telehealth and intelligent transportation has caused a rapid rise in demand for real-time communication and sensing \cite{6G_wireless_networks}. However, future wireless systems' throughput is typically constrained by increasingly crowded spectrum \cite{A vision}. In fact, ever since a long time ago, communication, sensing and other functions have been isolated and separated into different systems. For diversified, cross domain scenarios, multiple signal transmission and coordination will be required, leading to resource waste and response delays. To address this, researchers initially proposed the concept of radar-communication coexistence (RCC) \cite{Radar and}, which is a special case for integrated sensing and communications (ISAC). However, this approach, which employs orthogonal spectrum sharing, has limited effect on the improvement of both radar sensing and wireless communication. Therefore, radar-centered communication integration schemes have been proposed, aiming at transmitting radar information and communication signals through the same spectrum concurrently. As an instance, authors in \cite{Phase-modulation} proposed to transmit specific bits by modulating the phase and amplitude of the radar signal. However, since radar signals have a limited range of variance, this radar-centric communication integration faces challenges in terms of information transmission rate.
Meanwhile, the concept of sensor-assisted communication is gradually being introduced aiming at utilizing sensing information to improve the performance of wireless communication \cite{Sensor assisted wireless}. 
However, in sensor-assisted communication, the sensing function is still limited to some extent. In this context, ISAC becomes a viable option for realizing 6G wireless networks \cite{Integrating sensing}, \cite{Toward Multi-Functional}. Both sensing and communication performance can be improved by optimizing the wireless transceiver. In \cite{Integrating Sensing Computing}, a framework integrating sensing, computing and communication was proposed. Besides, two beamforming algorithms, derived from a multi-objective optimization framework, were also introduced. 

However, limitations in regional coverage and the speed of data processing can constrain these ground-based ISAC systems. In contrast to the terrestrial ISAC systems, satellite systems offer wide coverage, long-range communication capability, high flexibility, and relatively small transmission loss. These advantages make LEO satellites a perfect option for wide-area or even global ISAC systems. Yet, the scale of onboard hardware is constrained by the conventional LEO satellites' limited payload capacity. Besides, these satellites also suffer from performance degradation due to substantial Doppler shifts and propagation delays \cite{Ubiquitous_integrated}. Thankfully, the new-generation LEO satellites' improved payload processing capability allows for the simultaneous execution of multiple operations on just one hardware platform. Moreover, with the global expansion of LEO satellites, there has been a great improvement in their communication capability and sensing accuracy, both in time and space \cite{RSMA-based}.
For instance, the study in \cite{Beam_squint} was dedicated to ISAC technology's utilization in massive multiple-input multiple-output (MIMO) LEO satellite networks, providing wider coverage for ISAC compared to traditional ground-based ISAC methods.
The authors in \cite{Joint_Sensing} discussed the strategy of joint communication, sensing and compression for satellite-terrestrial networks, and jointly optimized the satellite's transmission power, the central processor frequency, the ratio of compression and the rate of sensing, proposing a highly efficient method to tackle the non-convex optimization problem.

Nevertheless, existing technologies of satellite ISAC are commonly based on perfect channel state information (CSI). However, obtaining the perfect CSI is impractical for fast-mobility LEO satellites \cite{MIMO_satellite}. In practice, noise, quantization errors, etc., can affect the accuracy of CSI. Besides, there is a delay during CSI acquisition, which will result in phase error in the CSI obtained by the satellite. Imperfect CSI may lead to QoS outage and other serious problems, causing performance degradation. Therefore, the imperfect CSI issue should be considered with robust algorithms designed to reduce its impact. 
In fact, in satellite communication, the issue of channel phase uncertainty has been partially studied. For example, in \cite{Proportionally_fair}, the authors proposed a robust beamforming scheme under phase uncertainty for multibeam multicast systems. Moreover, in \cite{Robust_Design}, the authors considered channel phase uncertainty and proposed robust beamforming algorithms for LEO satellites with energy constraints.
However, in the context of multi-user, multi-target ISAC, only a limited number of studies incorporate imperfect CSI. For instance, in \cite{Robust_Beamforming_for}, the authors considered target position uncertainty in high-mobility scenarios and modeled complex motion errors via particle filtering. In \cite{Robust_transceiver}, bounded estimation errors in the channel matrix were considered, and worst-case robust optimization problems were formulated. However, these studies do not account for the intricate cross-functional interference due to imperfect CSI. As mentioned before, in satellite ISAC systems, channel estimation error arises from noise, quantization errors, and acquisition delays. The phase component of the channel is particularly vulnerable to these impairments \cite{Robust_Design}. Thus, it is reasonable and critical to explicitly address such phase errors in robust designs for satellite ISAC.
However, directly extending existing robust designs to satellite ISAC systems for multi-user and multi-target scenarios will bring a series of problems. In such systems, serving multiple users and sensing multiple targets involves numerous transmission links, each characterized by stochastic phase errors with different variances. This complexity introduces a twofold challenge beyond isolated function performance degradation. First, the coupling between sensing and communication functions means that phase errors simultaneously degrade the estimation accuracy of target parameters and the quality of user-received signals. Second, these errors randomly exacerbate the intricate interference among all concurrent sensing and communication waveforms. Therefore, guaranteeing high-quality instantaneous service for both functions under such conditions remains a critical, yet unresolved, robust design problem.

Driven by this, we aim to establish a robust beamforming framework for ISAC system in the presence of channel phase uncertainty. We summarize our contributions as follows.
\begin{enumerate}
\item We put forward a framework for ISAC in LEO satellite systems considering the cross-functional interference exacerbated by channel phase uncertainty.

\item We derive the QoS metrics for communication and sensing, i.e., the mean squared error (MSE) and the signal-to-interference-plus-noise ratio (SINR), and construct an optimization problem to minimize the total transmit power, while meeting the QoS requirements.

\item We propose an iterative robust algorithm for ISAC in LEO satellite systems by alternately optimizing the receive and transmit beamforming vectors to minimize the transmit power and analyze key parameters such as QoS level, channel phase uncertainty level's influence on the performance of the algorithm.

\end{enumerate}

This paper is structured as follows: Section \uppercase\expandafter{\romannumeral2} presents a system model for ISAC in LEO satellite systems. Section \uppercase\expandafter{\romannumeral3} proposes a robust beamforming design algorithm to minimize the transmit power. Section \uppercase\expandafter{\romannumeral4} presents simulation results to verify the proposed algorithm's effectiveness. Finally, Section \uppercase\expandafter{\romannumeral5} concludes the paper.

\emph{Notations}: We use ordinary letters to represent scalars, bold upper (lower) letters to denote matrices (column vectors). ${\left(  \cdot  \right)^{\rm{T}}}$ and ${\left(  \cdot  \right)^{\rm{H}}}$ represent the transpose and the conjugate transpose. ${\mathop{\rm Rank}\nolimits} \left(  \cdot  \right)$ and ${\mathop{\operatorname{tr}}\nolimits} \left(  \cdot  \right)$ indicate the rank, and the trace of a matrix. $\left|  \cdot  \right|$ and $\left\|  \cdot  \right\|$ represent the absolute value and Euclidean norm, $\mathbb{E}\{ \cdot \}$ and $\odot $ represent expectation and the Hadamard product. $\operatorname{Re}\{ \cdot \}$ and $\operatorname{Im}\{ \cdot \}$ stand for the extraction of real and imaginary parts. ${\mathbb{R}^{m \times n}}$, $\mathbb{R}{_ + ^{m \times n}}$, and ${\mathbb{C}^{m \times n}}$ are the sets of $m \times n$ dimensional real, non-negative real and complex matrices, respectively. ${{\cal K}^K}$ and ${{\cal S}^K}$ represent the skew-symmetric and symmetric matrices. We use the notation $\operatorname{diag}({\bf{x}})$ to represent a square matrix which employs the elements of vector ${\bf{x}}$ as the principal diagonal entries, with all off-diagonal elements being zero, notation $\operatorname{blkdiag}(\mathbf{X}_1, \mathbf{X}_2, \cdots, \mathbf{X}_n)$ to represent a block diagonal matrix comprising $\mathbf{X}_1, \mathbf{X}_2, \cdots, \mathbf{X}_n$ as its principal diagonal blocks, with all off-block-diagonal elements being zero. ${\left[ {\bf{X}} \right]_{m,n}}$ represents matrix ${\bf{X}}$'s [m,n]-th element. ${\left(  \cdot  \right)^{\text{c}}}$, ${\left(  \cdot  \right)^{\text{s}}}$ and ${\left(  \cdot  \right)^{\rm{t}}}$ denote the variables associated with the channels from the satellite to CUs, from the satellite to sensing targets, and from the sensing targets to CUs, respectively.

\section{System model}

\begin{figure}
    \centering
    \includegraphics[width=1\linewidth]{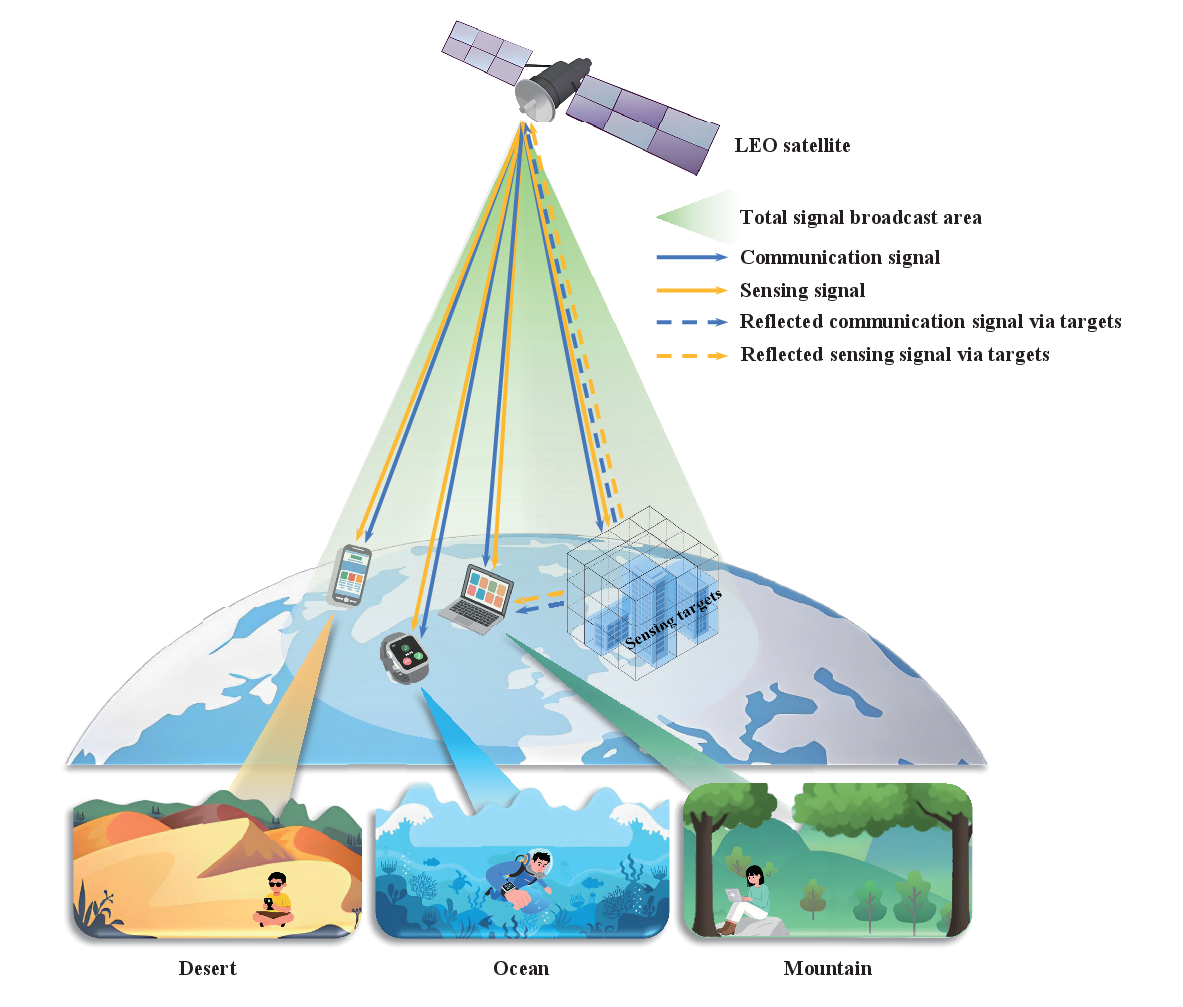}
    \caption{System model for ISAC in LEO satellite systems.}
    \label{fig:model}
\end{figure}

Consider an LEO satellite system integrating communication and sensing. As shown in Fig. \ref{fig:model}, the system is comprised of an LEO satellite equipped with $K$ antennas, $N$ single-antenna CUs, and a sensing area with targets inside. In particular, we divide the sensing area into $I$ blocks of the same size as resolvable sensing targets. The LEO satellite communicates with the CUs while performing target sensing. Specifically, the LEO satellite sends data information to multiple CUs and obtains the targets' reflection coefficients simultaneously over the same spectrum. To this end, at the start of a time slot, the LEO satellite broadcasts the ISAC signal as follows
\begin{equation}
{\bf{x}} = \underbrace {\sum\limits_{i = 1}^I {{{\bf{a}}_i} s_i^{{\text{sens}}}} }_{{\text{Sensing signal}}} + \underbrace {\sum\limits_{n = 1}^N {{{\bf{b}}_n} s_n^{{\text{comm}}}} }_{{\text{Communication signal}}},
\end{equation}
where $ {{\bf{a}}_i}\in \mathbb{C}{^{K \times 1}}$ represents the sensing beamforming vector used by the satellite to transmit the $i$-th sensing signal $s_i^{{\text{sens}}}$, which is used to detect the $i$-th sensing target. In addition, ${{\bf{b}}_n} \in \mathbb{C}{^{K \times 1}} $ represents the communication beamforming vector used by the satellite to transmit the $n$-th communication signal $s_n^{{\text{comm}}}$, which is supposed to be transmitted to the $n$-th CU. For ease of analysis, we consider the signals to be independent of each other and obey a Gaussian distribution with unit norm, i.e.,
\begin{equation}
\begin{array}{l}
\mathbb{E}\left( {s_i^{{\text{sens}}}\left(s_j^{{\text{sens}}}\right)^{\rm{H}}} \right) = \left\{ \begin{array}{l}
0,i \ne j\\
1,i = j
\end{array} \right.{\rm{\ ,      }}\\
{\rm{ }}\mathbb{E}\left( {s_n^{{\text{comm}}}\left(s_m^{{\text{comm}}}\right)^{\rm{H}}} \right) = \left\{ \begin{array}{l}
0,n \ne m\\
1,n = m
\end{array} \right.{\rm{\ ,      }}\\
\mathbb{E}\left( {s_i^{{\text{sens}}}\left(s_n^{{\text{comm}}}\right)^{\rm{H}}} \right) = 0 {\rm{\ .      }}
\end{array}
\end{equation}
In the following, we formulate the channel, target sensing and communication models, respectively.

\subsection{Channel Model}
Considering the CUs are located on the ground, the satellite-ground channel between the $n$-th CU and satellite is modeled as  \cite{Multiple-Satellite}
\begin{equation}
{{\bf{h}}_n} = \sqrt {C_n^{\text{c}}} {(\boldsymbol{\beta}_n^{\text{c}})}^{1/2} \odot {(\boldsymbol{\rho }_n^{\text{c}})}^{1/2} \odot {{\boldsymbol{\ell}}^{\text{c}}} ,
\end{equation}
where ${C_n^{\text{c}}}$ is the large-scale fading, ${{\boldsymbol{\beta }}_n^{\text{c}}}$ is the satellite antenna gain, ${{\boldsymbol{\rho }}_n^{\text{c}}}$ represents the rain fading vector and ${\boldsymbol{\ell}}^{\text{c}}$ represents the small-scale fading.
Specifically, ${C_n^{\text{c}}}$ can be expressed as
\begin{equation}
{C_n^{\text{c}}} = {\left( {\frac{c}{{4\pi f{d_n}}}} \right)^2}\frac{{{G_n}}}{\sigma_{n}^2}, 
\end{equation}
where $f$ and $c$ respectively represent the signal frequency and the light speed. In addition, $\sigma_{n}^2$ is the variance of the additive white Gaussian noise (AWGN) $n_n$ at the $n$-th CU, satisfying $\sigma_{n}^2={\kappa BT_n}$. Herein, $\kappa, B, T_n$ represent Boltzmann constant, channel bandwidth and noise temperature at the $n$-th CU, respectively. Furthermore, ${d_{n}}$ and ${G_n}$ represent the distance from the $n$-th CU to satellite, and the receiving antenna gain of the $n$-th CU, respectively. Given that LEO satellites typically maintain moderate elevation angles over the region, the resulting dominant line-of-sight (LOS) component justifies the use of the Rician fading model.
Therefore, the small-scale fading vector $\boldsymbol{\ell^{\text{c}}}$ can be given by 
\begin{equation}
\boldsymbol{\ell^{\text{c}}}  = \left( {\sqrt {\frac{{\lambda^{\text{c}} _n}}{{{{\lambda^{\text{c}} _n}} + 1}}} {{\bf{h}}_n^{\text{LOS}}}  + \sqrt {\frac{1}{{{\lambda^{\text{c}} _n} + 1}}} {{\bf{h}}_n^{\text{NLOS}}}} \right),
\end{equation}
where 
${{\lambda^{\text{c}} _n}}$
is the Rician factor between the satellite and the $n$-th CU, 
${\bf{h}}_n^{\text{LOS}}$ and ${\bf{h}}_n^{\text{NLOS}}$
indicate the LOS and non-line-of-sight (NLOS) parts of the satellite-ground channel.
Moreover, the element of the satellite antenna gain vector ${{\boldsymbol{\beta }}_n^{\text{c}}}$ can be expressed as
\begin{equation}
{\left[ {{{\boldsymbol{\beta }}_n^{\text{c}}}} \right]_i} = {M}{\left( {\frac{{{J_1}\left( {{\upsilon _i}} \right)}}{{2{\upsilon _i}}} + 36\frac{{{J_3}\left( {{\upsilon _i}} \right)}}{{\upsilon _i^3}}} \right)^3},
\end{equation}
where ${\upsilon _i} = 2.071\left( {{{\sin \left( {{\varphi _{i,n}}} \right)} \mathord{\left/
 {\vphantom {{\sin \left( {{\varphi _{i,n}}} \right)} {\sin \left( {{\varphi ^{3dB}}} \right)}}} \right.\kern-\nulldelimiterspace} {\sin \left( {\varphi ^{3dB}} \right)}}} \right)$ with $M$, ${\varphi _{i,n}}$ and ${\varphi ^{3dB}}$ being the maximum satellite antenna gain, the angle between the satellite's $i$-th antenna and $n$-th CU, the 3dB angle of the satellite, respectively.
In addition, ${{\boldsymbol{\rho }}_n^{\text{c}}}$ is the rain fading vector, whose dB form $({\boldsymbol{\rho }_n^{\text{c}}})^{dB} = 20\lg {{\boldsymbol{\rho }}_n^{\text{c}}}$
follows a lognormal distribution $\ln \left( {({\boldsymbol{\rho }_n^{\text{c}}})^{dB}} \right) \sim {\cal C}{\cal N}\left( {{\mu _{{{{\rho}}_n^{\text{c}}}}},{\mkern 1mu} \sigma _{{{\rho}_n^{\text{c}}}}^2} \right)$ \cite{Joint_beamforming}. The satellite channel's amplitude is governed by antenna gain, small-scale fading, large-scale fading and rain fading, all of which can be regarded as approximately constant within a short time slot. In contrast, the channel phase is highly time-varying due to factors such as satellite mobility and Doppler shifts, and thus changes much faster than the amplitude \cite{Robust_Design}. Therefore, the satellite may obtain outdated channel phase information, resulting in channel phase uncertainty. Specifically, the phase of ${{\bf{h}}_n}$ is given by 
\begin{equation}
{{\boldsymbol{\theta }}_n^{\text{c}}} =  {{\hat{\boldsymbol{\theta }}_n^{\text{c}}}}  + {{\bf{e}}_n^{\text{c}}},
\end{equation}
where ${{\boldsymbol{\theta }}_n^{\text{c}}}$ represents the phase shift of ${{\bf{h}}_n}$, which is independent of each other, $ {\hat{\boldsymbol{\theta }}_n^{\text{c}}}$ is the obtained channel phase at the $n$-th CU and ${{\bf{e}}_n^{\text{c}}}$ denotes the phase error, following an independent and identically 
distributed Gaussian distribution, i.e., ${{\bf{e}}_n^{\text{c}}} \sim {\cal N}({\bf{0}},({\sigma _n^{\text{c}}})^{2}{{\bf{S}}_n^{\text{c}}})$. Herein, ${{\bf{S}}_n^{\text{c}}}$ and $({\sigma _n^{\text{c}}})^{2}$ are the normalized covariance matrix and the variance of phase error, respectively. Then, the actual channel ${\bf{h}}_n$ can be rewritten as
\begin{equation}
{{\bf{h}}_n} =  {{\hat{\bf{h}}_n}}  \odot {{\bf{q}}_n^{\text{c}}} = \operatorname{diag}( { {{\hat{\bf{h}}_n}} } ){{\bf{q}}_n^{\text{c}}},
\end{equation}
where ${{\bf{q}}_n^{\text{c}}} = \exp \{ j{{\bf{e}}_n^{\text{c}}}\}$ and the obtained channel $\hat{\bf{h}}_n$ can be expressed as
\begin{equation}
{{\hat{\bf{h}}_n}} = \sqrt {{C_n^{\text{c}}}} ({{\boldsymbol{\beta }}_n^{\text{c}}})^{1/2} \odot ({\boldsymbol{\rho }_n^{\text{c}}})^{1/2} \odot \left|  {\boldsymbol{\ell^{\text{c}}}}  \right| \odot \exp \{ j{{\hat{\boldsymbol{\theta }}_n^{\text{c}}}}\}.
\end{equation}

Similarly, the actual channel ${\bf{g}}_{i}$ between the $i$-th sensing target and satellite is modeled as
\begin{equation}
\label{channel_g}
{{\bf{g}}_{i}} =  {{\hat{{\bf{g}}}_{i}}}  \odot {{\bf{q}}_i^{\text{s}}} = \operatorname{diag}\left( { {{\hat{{\bf{g}}}_i}} } \right){{\bf{q}}_i^{\text{s}}},
\end{equation}
where ${{\hat{{\bf{g}}}_{i}}}$ is the obtained channel, which can be given by
\begin{equation}
{{\hat{{\bf{g}}}_{i}}} = \sqrt {C_i^{\text{s}}} ({\boldsymbol{\beta}_i^{\text{s}}})^{1/2} \odot ({\boldsymbol{\rho }_i^{\text{s}}})^{1/2} \odot \left|{\boldsymbol{\ell} ^{\text{s}}}\right|\odot \exp \{ j{\hat{\boldsymbol{\theta}}^{\text{s}}_i}\},
\end{equation}
Herein, ${C_i^{\text{s}}}$ represents the large-scale fading factor and ${\boldsymbol{\ell} ^{\text{s}}}$ denotes the small scale fading, given by
\begin{equation}
{\boldsymbol{\ell} ^{\text{s}}} = \sqrt {\frac{\lambda _i^{\text{s}}}{{{\lambda _i^{\text{s}}} + 1}}} {{\bf{g}}_i^\text{LOS}}  + \sqrt {\frac{1}{{{\lambda _i^{\text{s}}} + 1}}} {{\bf{g}}_i^\text{NLOS}},
\end{equation}
where ${\lambda _i^{\text{s}}}$, ${{\bf{g}}_i^\text{LOS}}$ and ${{\bf{g}}_i^\text{NLOS}}$
indicate the Rician factor, LOS and NLOS links between the satellite and the $i$-th sensing target. Besides, ${\boldsymbol{\beta}_i^{\text{s}}}$, ${{\boldsymbol{\rho }}_i^{\text{s}}}$ and ${\hat{\boldsymbol{\theta}}^{\text{s}}_i}$ represent the satellite antenna gain, the rain fading vector whose dB form follows a lognormal distribution $\ln \left( {({{\boldsymbol{\rho }}_i^{\text{s}}})^{dB}} \right) \sim {\cal C}{\cal N}\left( {{\mu _{{{\rho}}_i^{\text{s}}}},{\mkern 1mu} \sigma _{{\rho}_i^{\text{s}}}^{2}} \right)$ and the obtained channel phase shift, respectively.
In addition, ${{\bf{q}}_i^{\text{s}}} = \exp \{ j{{\bf{e}}_i^{\text{s}}}\}$, where ${{\bf{e}}_i^{\text{s}}}$ represents the phase error of ${{\bf{g}}_{i}}$ and ${{\bf{e}}_i^{\text{s}}} \sim {\cal N}({\bf{0}},({\sigma _i^{\text{s}}})^{2}{{\bf{S}}_i^{\text{s}}})$. Herein, ${{\bf{S}}_i^{\text{s}}}$ is the normalized covariance matrix and $({\sigma _i^{\text{s}}})^{2}$ is the variance of phase error, respectively.

Signals transmitted from the satellite to both CUs and sensing targets experience severe Doppler shifts and propagation delays. Considering the deterministic nature of the satellite orbit, the Doppler shifts and propagation delays caused by the high-speed satellite motion can be effectively pre-compensated. Specifically, the state evolution model in \cite{Exploiting_on-orbit} is employed to pre-calculate these shifts based on real-time GNSS data and satellite ephemeris. Pre-compensation is then directly applied at the satellite transmitter by adaptively tuning the carrier frequency and transmit timing. Consequently, after compensating for the deterministic Doppler and delay components arising from the satellite's movement, the residual offsets are primarily attributed to the motion of the targets. The treatment of these target-induced effects will be discussed in the subsequent sensing model.

For the terrestrial channel between the sensing targets and the CUs, the LOS component can be reasonably neglected, thereby reducing the Rician fading model to a Rayleigh fading model. Hence, the terrestrial channel from the $i$-th sensing target to the $n$-th CU can be expressed as
\begin{equation}
{p_{i,n}} = {{\hat{p}_{i,n}}} {o_{i,n}},
\end{equation}
where ${\hat{p}_{i,n}}$ denotes the obtained channel, which is expressed as
\begin{equation}
{\hat{p}_{i,n}} = \sqrt {C_{i,n}^{\text{t}}} ({\rho _{i,n}^{\text{t}}})^{1/2}\left|{p_{i,n}^{\text{NLOS}}}\right|\exp \{ j{\hat{\theta} _{i,n}^{\text{t}}}\}.
\end{equation}
Herein, ${C_{i,n}^{\text{t}}}$ is the large-scale fading factor and ${\rho _{i,n}^{\text{t}}}$ represents the rain fading scalar whose dB form also follows a lognormal distribution $\ln \left( {({\rho _{i,n}^{\text{t}}})^{dB}} \right) \sim {\cal C}{\cal N}\left( {{\mu _{{\rho _{i,n}^{\text{t}}}}},{\mkern 1mu} \sigma _{{\rho _{i,n}^{\text{t}}}}^2} \right)$. In addition, ${o_{i,n}} = \exp \{ j{e_{i,n}^{\text{t}}}\}$, where ${e_{i,n}^{\text{t}}}$ denotes the phase error of channel ${p_{i,n}}$ and also follows normal random distribution with mean 0 and variance $({\sigma _{i,n}^{\text{t}}})^{2}$.

\subsection{Target Sensing Model}
For target sensing, the satellite receives the reflected signal carrying sensing target's information and employs a dedicated sensing receiver to estimate the target's reflection coefficient, facilitating a wide range of practical sensing services. For instance, in agricultural Internet of Things (IoT) systems, the estimated reflection coefficients are essential for characterizing land surface properties, such as soil moisture and surface roughness \cite{Estimation_of}, which provide critical data for precision irrigation. Furthermore, for intelligent surveillance applications, these parameters enable reliable target classification to identify specific objects \cite{Optimum_power}, thereby enhancing regional security.

Unlike the predictable satellite platform, sensing targets may exhibit non-cooperative motion, which introduces additional Doppler shifts and propagation delays to the reflected signal. Let $\nu_i^{s}$ and $\tau_i^{s}$ denote the Doppler shift and propagation delay associated with the $i$-th target, respectively. The continuous-time received signal at the sensing receiver can be expressed as
\begin{equation}
{\bf{w}}(t) = \underbrace{\sum\limits_{i = 1}^I {{r_i} {{\bf{g}}_i}{\bf{g}}_i^{\rm{H}}{\bf{x}}(t - \tau_i^{s})} e^{j2\pi \nu_i^{s} t}}_{{\text{Signals reflected by targets}}} + \underbrace{{{\bf{n}}^{\prime}}(t)}_{{\text{Noise}}},
\end{equation}
where ${\bf{x}}(t)$ represents the continuous-time transmitted signal vector corresponding to the symbol vector ${\bf{x}}$. 

To address this issue, Doppler shift and delay estimation and compensation methods—such as Kalman filtering \cite{Application_of_Kalman}, and maximum likelihood estimation \cite{Maximum_likelihood}—can be employed at the sensing receiver to mitigate the effects of target movement. Furthermore, although the major Doppler and propagation delay components are compensated, the target motion inevitably leaves residual timing offsets and Doppler-induced phase rotations, which can be incorporated into our channel phase uncertainty modeling. Therefore, by omitting the time index for ease of presentation, the received compensated signal at the satellite is concisely given by
\begin{equation}
{\bf{w}} = \underbrace {\sum\limits_{i = 1}^I {{r_i}{{\bf{g}}_i}{\bf{g}}_i^{\rm{H}}{\bf{x}}} }_{{\text{Signals reflected by targets}}} + \underbrace {{\bf{n}}^{\prime}}_{{\text{Noise}}},
\end{equation}
where $r_i$ and ${\bf{n}}^{\prime}$ represent the $i$-th target's reflection coefficient and the AWGN with variance $(\sigma^{\prime})^2$. Then, the satellite performs receive beamforming to enhance the desired signal while suppressing interference from other reflected signals, and the resulting estimate of the $i$-th target's reflection coefficient is given by
\begin{align}
\widehat{r}_i &= \mathbf{v}_i^{\rm{H}}\mathbf{w} \notag \\
&= \mathbf{v}_i^{\rm{H}}\sum_{i = 1}^I r_i\mathbf{g}_i\mathbf{g}_i^{\rm{H}}\mathbf{x} + \mathbf{v}_i^{\rm{H}}{\bf{n}}^{\prime} \notag \\
&= r_i\left( \mathbf{v}_i^{\rm{H}}\mathbf{g}_i\mathbf{g}_i^{\rm{H}}\left( \sum_{l = 1}^I \mathbf{a}_l s_l^{\text{sens}} \right) + \mathbf{v}_i^{\rm{H}}\mathbf{g}_i\mathbf{g}_i^{\rm{H}}\sum_{n = 1}^N \mathbf{b}_n s_n^{\text{comm}} \right) \notag \\
&\quad + \mathbf{v}_i^{\rm{H}}\sum_{\substack{j = 1 \\ j \ne i}}^I r_j\mathbf{g}_j\mathbf{g}_j^{\rm{H}}\sum_{l = 1}^I \mathbf{a}_l s_l^{\text{sens}} \notag \\
&\quad + \mathbf{v}_i^{\rm{H}}\sum_{\substack{j = 1 \\ j \ne i}}^I r_j\mathbf{g}_j\mathbf{g}_j^{\rm{H}}\sum_{n = 1}^N \mathbf{b}_n s_n^{\text{comm}} \notag \\
&\quad + \mathbf{v}_i^{\rm{H}}{\bf{n}}^{\prime},
\end{align}
where ${\bf{v}}_i^{\rm{H}}$ is a $K$-dimensional sensing receive beamforming vector, which is a function related to the $i$-th sensing signal, i.e. ${{\bf{v}}_i} = {{\bf{v}}_i^{\#}}s_i^{\text{sens}}$, where ${{\bf{v}}_i^{\#}}$ is irrelevant to the sensing signal.

Since the estimate of the target reflection coefficient is distorted by interference from other sensing and communication signals in addition to noise, we adopt the MSE as the metric for the performance of sensing, and a minimum MSE (MMSE) receiver is designed to mitigate such interference. In particular, the MSE between ${\widehat r_i}$ and ${r_i}$ is given in equation (\ref{MSE_E}) on the top of the following page.
\begin{figure*}[!t]
\centering
\begin{equation}
\label{MSE_E}
\begin{aligned}[b]
{\operatorname{MSE}}_i^{\text{sens}} &= \mathbb{E}\left\{ (\widehat{r_i} - r_i)(\widehat{r_i} - r_i)^{\rm{H}} \right\} \\
&= \mathbb{E}\Biggl\{ ({{\bf{v}}_i^{\#}})^{\rm{H}} \biggl[ \sum_{j=1}^I {R}_j^2 \mathbf{g}_j \mathbf{g}_j^{\rm{H}} \biggl(  \sum_{l=1}^I \mathbf{a}_l\mathbf{a}_l^{\rm{H}} + \sum_{n=1}^N  \mathbf{b}_n\mathbf{b}_n^{\rm{H}} \biggr) \mathbf{g}_j \mathbf{g}_j^{\rm{H}} \biggr] {\bf{v}}_i^{\#} \\
&\quad -  {R}_i^2 \mathbf{a}_i^{\rm{H}} \mathbf{g}_i \mathbf{g}_i^{\rm{H}} {\bf{v}}_i^{\#} -  {R}_i^2 ({{\bf{v}}_i^{\#}})^{\rm{H}} \mathbf{g}_i \mathbf{g}_i^{\rm{H}} \mathbf{a}_i + {R}_i^2 + (\sigma^{\prime})^2 \| ({{\bf{v}}_i^{\#}})^{\rm{H}} \|^2 \Biggr\}.
\end{aligned}
\end{equation}
\hrulefill
\vspace*{4pt}
\end{figure*}
Herein, ${{R}_i}$ represents the actual reflection coefficient ${r_i}$'s root mean squared (RMS) value and can be expressed as
\begin{equation}
{{{R}}_i} = \sqrt {\sum\limits_{j \in {\Omega _j}} {{\pi _j}{{\left| {{r_j}} \right|}^2}} },
\end{equation}
where ${\pi _j} \in \mathbb{R}{_ + ^{1 \times 1} }$ denotes the prior probabilities of different types of sensing targets \cite{Optimum_power}. We assume that ${{{R}}_i}$ can be obtained in advance. Next, we introduce ${{\bf{G}}_i} = {{\bf{g}}_{i}}{\bf{g}}_{i}^{\rm{H}}$ and ${{\bf{Q}}_i^{\text{s}}} = \mathbb{E}\left( {{{\bf{q}}_i^{\text{s}}}({{\bf{q}}_i^{\text{s}}})^{{\rm{H}}}} \right)$ based on the model of channel ${{\bf{g}}_{i}}$ which is provided in equation (\ref{channel_g}) and simplify it as
\begin{equation}
{{\bf{G}}_{i}} = \operatorname{diag}\left( { {{\hat{\bf{g}}_i}} } \right){{\bf{Q}}_i^{\text{s}}}\operatorname{diag}\left( {{ \hat{\bf{g}}_i ^{\rm{H}}}} \right).
\end{equation}
According to the definition of the phase error ${{\bf{e}}_i^{\text{s}}}$, the element of ${{\bf{Q}}_i^{\text{s}}}$ is calculated as
\begin{equation}
\left[ {\mathbf{Q}_i^{\text{s}}} \right]_{x,y} = 
\begin{cases}
    1, & \text{if } x = y \\
    \exp \left( - j({\sigma_i^{\text{s}}})^{2} \right), & \text{otherwise}
\end{cases}.
\end{equation}
Accordingly, the MSE of target sensing in equation (\ref{MSE_E}) can be simplified as
\begin{equation}
\label{MSE_def}
\begin{aligned}[b]
{\operatorname{MSE}}_i^{\text{sens}} 
=& ({{\bf{v}}_i^{\#}})^{\rm{H}} \Biggl[ \sum_{j=1}^I {{R}}_j^2
    \operatorname{diag}\left( {\hat{\bf{g}}_i} \right) {{\bf{Q}}_i^{\text{s}}} \operatorname{diag}\left( \hat{\bf{g}}_i^{\rm{H}} \right) \\
& \cdot \biggl(  \sum_{l=1}^I {\bf{a}}_l{\bf{a}}_l^{\rm{H}} + \sum_{n=1}^N {\bf{b}}_n{\bf{b}}_n^{\rm{H}} \biggr) \\
& \cdot \operatorname{diag}\left( {\hat{\bf{g}}_i} \right) {{\bf{Q}}_i^{\text{s}}} \operatorname{diag}\left( \hat{\bf{g}}_i^{\rm{H}} \right) \Biggr] {\bf{v}}_i^{\#} \\
&-  {{R}}_i^2 {\bf{a}}_i^{\rm{H}} \operatorname{diag}\left( {\hat{\bf{g}}_i} \right) {{\bf{Q}}_i^{\text{s}}} \operatorname{diag}\left( \hat{\bf{g}}_i^{\rm{H}} \right) {\bf{v}}_i^{\#} \\
&-  {{R}}_i^2 ({{\bf{v}}_i^{\#}})^{\rm{H}} \operatorname{diag}\left( {\hat{\bf{g}}_i} \right) {{\bf{Q}}_i^{\text{s}}} \operatorname{diag}\left( \hat{\bf{g}}_i^{\rm{H}} \right) {\bf{a}}_i \\
&+ {{R}}_i^2 + (\sigma^{\prime})^2 \| ({{\bf{v}}_i^{\#}})^{\rm{H}} \|^2.
\end{aligned}
\end{equation}

\subsection{Communication Model}
The signal received by the $n$-th CU comprises directly transmitted signals from the satellite, signals reflected by sensing targets and noise, which can be expressed as
\begin{equation}
\label{received_signal}
\begin{aligned}[b]
{y}_n^{\text{comm}} = & \underbrace {{\bf{h}}_n^{\rm{H}}{\bf{x}}}_{{\text{Direct\ signal}}} + \underbrace {\sum\limits_{i = 1}^I {{r_i}{p_{i,n}}{\bf{g}}_i^{\rm{H}}{\bf{x}}} }_{{\text{Signals\ reflected\ by\ targets}}} + \underbrace {{{n}_n}}_{{\text{Noise}}}\\
 =& \underbrace {({\bf{h}}_n^{\rm{H}}+{\sum\limits_{j = 1}^I {{r_j}{p_{j,n}}{\bf{g}}_j^{\rm{H}} }}){{\bf{b}}_n} s_n^{\text{comm}}}_{{\text{Desired signal}}}\\ 
 &+ \underbrace {({\bf{h}}_n^{\rm{H}}+{\sum\limits_{j = 1}^I {{r_j}{p_{j,n}}{\bf{g}}_j^{\rm{H}} }})\sum\limits_{m = 1,m \ne n}^N {{{\bf{b}}_m} s_m^{\text{comm}}} }_{{\text{Inter-CU interference}}}\\
 &+ \underbrace {({\bf{h}}_n^{\rm{H}}+{\sum\limits_{j = 1}^I {{r_j}{p_{j,n}}{\bf{g}}_j^{\rm{H}} }})\sum\limits_{i = 1}^I {{{\bf{a}}_i} s_i^{\text{sens}}} }_{{\text{Sensing interference}}} + {{n}_n}.
\end{aligned}
\end{equation}
It can be seen from (\ref{received_signal}) that the desired signal comprises not only the corresponding communication signal directly transmitted from the satellite to the CU, but also that reflected by the targets and then received by the CU. To assess communication performance, the SINR is employed. According to (\ref{received_signal}), the SINR at the $n$-th CU can be given by
\begin{equation}
\label{SINR}
\Gamma_n = \dfrac{
    \left|  \mathbf{h}_n^{\rm{H}} \mathbf{b}_n  \right|^2 + \sum_{j = 1}^I  {\left|  r_j p_{j,n} {\bf{g}}_j^{\rm{H}} {\bf{b}}_n \right|^2}
}{
    \sum\limits_{\substack{m=1 \\ m \neq n}}^N \!\!\! \left| \mathbf{h}_n^{\rm{H}} \mathbf{b}_m  \right|^2 \!+\! 
    \sum\limits_{i=1}^I \! \left| \mathbf{h}_n^{\rm{H}} \mathbf{a}_i  \right|^2 \!+\! 
    {X}_n \!+\! 
    \sigma_n^2
},
\end{equation}
where
\begin{equation}
\begin{aligned}[b]
{X}_n = 
& \sum_{j = 1}^I \sum_{i = 1}^I {\left| r_j p_{j,n} {\bf{g}}_j^{\rm{H}} {\bf{a}}_i  \right|^2} \\ 
& + \sum_{j = 1}^I \sum_{\substack{m=1 \\ m \neq n}}^N {\left|  r_j p_{j,n} {\bf{g}}_j^{\rm{H}} {\bf{b}}_m \right|^2}.
\end{aligned}
\end{equation}

From equation (\ref{MSE_def}) and (\ref{SINR}) above, it can be found that the ISAC in LEO satellite systems' performance critically depends on the transmit beamforming vectors ${{\bf{a}}_i}$, ${{\bf{b}}_n}$ and the receive beamforming vector ${{\bf{v}}_i}$.
Therefore, this paper seeks to create a unified framework for ISAC in LEO satellite systems to enhance both sensing and communication functions' performance by jointly optimizing receive and transmit beamforming. 

\section{Robust Beamforming Design for ISAC in LEO Satellite Systems}
In this section, we propose a robust beamforming design scheme for ISAC in LEO satellite systems considering channel phase uncertainty. In particular, the proposed design minimizes the total transmit power while meeting both communication and sensing requirements. Due to the phase uncertainty, an outage probability constraint is applied to satisfy communication requirements. Consequently, the overall optimization problem can be formulated as                                           
\begin{subequations} 
\label{M1}
\begin{align}
\mathop {\min }\limits_{{{\bf{a}}_i},{{\bf{b}}_n},{{\bf{v}}_i^{\#}}} 
    &\sum\limits_{i = 1}^I {{{\left\| {{{\bf{a}}_i}} \right\|}^2}}  + \sum\limits_{n = 1}^N {{{\left\| {{{\bf{b}}_n}} \right\|}^2}} \label{M1_objective_function} \\
  \text{s.t.}\
    &\operatorname{MSE}_i^{\text{sens}} \le {\delta _i}, \label{M1_MSE_constraint}\\
    &\Pr \left\{ {{\Gamma _n} \ge {\gamma _n}} \right\} \ge 1 - {\varrho_n}, \label{M1_SINR_constraint}
\end{align}
\end{subequations}
where the objective function (\ref{M1_objective_function}) is minimizing the satellite's overall transmit power, constraint (\ref{M1_MSE_constraint}) derived from (\ref{MSE_def}) signifies the QoS requirement for sensing with $\delta _i$ being the maximum tolerance of MSE for the $i$-th sensing target, and constraint (\ref{M1_SINR_constraint}) based on (\ref{SINR}) is the SINR outage probability constraint, which imposes demands upon the quality of communication between the LEO satellite and the CUs it serves with $\Gamma _n$ and $\varrho_n$ being the required minimum SINR and the SINR outage probability threshold for the $n$-th CU, respectively. 

Due to its non-convexity, the formulated optimization problem cannot be solved directly. To address this issue, we introduce two auxiliary variables ${{\bf{A}}_i} = {{\bf{a}}_i}{\bf{a}}_i^{\rm{H}}$, ${{\bf{B}}_n} = {{\bf{b}}_n}{\bf{b}}_n^{\rm{H}}$. 
Then, equation (\ref{SINR}) can be reformulated as
\begin{equation}
\label{SINR_SDR}
\Gamma_n =  \dfrac{
     \operatorname{tr} \! \left( \mathbf{h}_n^{\rm{H}} \mathbf{B}_n \mathbf{h}_n \right) + \sum_{j = 1}^I |r_j p_{j,n}|^2   \operatorname{tr}({\bf{g}}_j^{\rm{H}} {\bf{B}}_n {\bf{g}}_j)
}{
    \sum\limits_{\substack{m=1 \\ m \neq n}}^N \!\!\!  \operatorname{tr} \! \left( \mathbf{h}_n^{\rm{H}} \mathbf{B}_m \mathbf{h}_n \right) \!+\! 
     \sum\limits_{i=1}^I \! \operatorname{tr} \! \left( \mathbf{h}_n^{\rm{H}} \mathbf{A}_i \mathbf{h}_n \right) \!+\! 
    {X}_n^{\prime} \!+\! 
    \sigma_n^2
},
\end{equation}
with
\begin{equation}
\label{Xnprime}
\begin{aligned}[b]
{X}_n^{\prime} = 
&  \sum_{j = 1}^I |r_j p_{j,n}|^2 \sum_{i = 1}^I \operatorname{tr}({\bf{g}}_j^{\rm{H}} {\bf{A}}_i {\bf{g}}_j) \\ 
& + \sum_{j = 1}^I |r_j p_{j,n}|^2 \sum_{\substack{m=1 \\ m \neq n}}^N  \operatorname{tr}({\bf{g}}_j^{\rm{H}} {\bf{B}}_m {\bf{g}}_j).
\end{aligned}
\end{equation}
Furthermore, equation (\ref{MSE_def}) can be reformulated as
\begin{equation}
\label{MSE_SDR}
\begin{aligned}[b]
{\operatorname{MSE}}_i^{\text{sens}} 
=& ({{\bf{v}}_i^{\#}})^{\rm{H}} \Biggl[ \sum_{j=1}^I {{R}}_j^2
    \operatorname{diag}\left( {\hat{\bf{g}}_i} \right) {{\bf{Q}}_i^{\text{s}}} \operatorname{diag}\left( {\hat{\bf{g}}_i}^{\rm{H}} \right) \\
& \cdot \biggl(  \sum_{l=1}^I {\bf{A}}_l + \sum_{n=1}^N {\bf{B}}_n \biggr) \\
& \cdot \operatorname{diag}\left( {\hat{\bf{g}}_i} \right) {{\bf{Q}}_i^{\text{s}}} \operatorname{diag}\left( {\hat{\bf{g}}_i}^{\rm{H}} \right) \Biggr] {\bf{v}}_i^{\#} \\
&-  {{R}}_i^2 {\bf{a}}_i^{\rm{H}} \operatorname{diag}\left( {\hat{\bf{g}}_i} \right) {{\bf{Q}}_i^{\text{s}}} \operatorname{diag}\left( {\hat{\bf{g}}_i}^{\rm{H}} \right) {\bf{v}}_i^{\#} \\
&-  {{R}}_i^2 ({{\bf{v}}_i^{\#}})^{\rm{H}} \operatorname{diag}\left( {\hat{\bf{g}}_i} \right) {{\bf{Q}}_i^{\text{s}}} \operatorname{diag}\left( {\hat{\bf{g}}_i}^{\rm{H}} \right) {\bf{a}}_i \\
&+ {{R}}_i^2 + (\sigma^{\prime})^2 \| ({{\bf{v}}_i^{\#}})^{\rm{H}} \|^2.
\end{aligned}
\end{equation}
However, the presence of ${{\bf{a}}_i}$ in equation (\ref{MSE_SDR}) is coupled with the auxiliary variable ${{\bf{A}}_i}$, which makes the MSE constraint non-convex. To address this issue, we use approximations to avoid the appearance of the first-order term containing ${{\bf{a}}_i}$ in the MSE. Specifically, by adopting an inequality ${\left| {{\bf{Y}} - 1} \right|^2} \le \left|{{\bf{Y}}^2} - 1\right| \ \text{while}\ \mathop{\rm Re}\nolimits({\bf{Y}}) \geq 0$, the MSE can be converted as (\ref{MSE_inequality}) on the top of the following page. 
\begin{figure*}[!t]
\centering
\begin{equation}
\label{MSE_inequality}
\begin{aligned}[b]
\operatorname{MSE}_i^{\text{sens}} 
&\leq \widetilde{\operatorname{MSE}}_i^{\text{sens}}=\Biggl| ({{\bf{v}}_i^{\#}})^{\rm{H}}\operatorname{diag}\left( {\hat{\bf{g}}_i} \right){{\bf{Q}}_i^{\text{s}}}\operatorname{diag}\left( {\hat{\bf{g}}_i}^{\rm{H}} \right)
    \biggl( \sum_{l=1}^{I} {\bf{A}}_l + \sum_{n=1}^{N} {\bf{B}}_n \biggr)
    \operatorname{diag}\left( {\hat{\bf{g}}_i} \right){{\bf{Q}}_i^{\text{s}}}\operatorname{diag}\left( {\hat{\bf{g}}_i}^{\rm{H}} \right){\bf{v}}_i^{\#} - 1 \Biggr| R_i^2 \\
&\quad + \sum_{\substack{j=1 \\ j \ne i}}^{I} ({{\bf{v}}_i^{\#}})^{\rm{H}}\operatorname{diag}\left( {\hat{\bf{g}}_j} \right){{\bf{Q}}_j^{\text{s}}}\operatorname{diag}\left( {\hat{\bf{g}}_j}^{\rm{H}} \right)
    \biggl(\sum_{l=1}^{I} {\bf{A}}_l + \sum_{n=1}^{N} {\bf{B}}_n \biggr)
    \operatorname{diag}\left( {\hat{\bf{g}}_j} \right){{\bf{Q}}_j^{\text{s}}}\operatorname{diag}\left( {\hat{\bf{g}}_j}^{\rm{H}} \right){\bf{v}}_i^{\#} R_j^2 
    + (\sigma^{\prime})^2 \| ({{\bf{v}}_i^{\#}})^{\rm{H}} \|^2.
\end{aligned}
\end{equation}
\hrulefill
\vspace*{4pt}
\end{figure*}

\begin{IEEEproof}
    Please refer to appendix A.
\end{IEEEproof}
By substituting equation (\ref{SINR_SDR}) and equation (\ref{MSE_inequality}) into (\ref{M1}), the original problem is rewritten as
\begin{subequations} 
\label{M3}
\begin{align}
  \mathop{\min}\limits_{\mathbf{A}_i,\mathbf{B}_n,{{\bf{v}}_i^{\#}}} 
    & \sum_{i=1}^I \operatorname{tr}(\mathbf{A}_i) + \sum_{n=1}^N \operatorname{tr}(\mathbf{B}_n) \\
  \text{s.t.}\ 
    & (\ref{M1_SINR_constraint}), \nonumber \\
    &\widetilde{\operatorname{MSE}}_i^{\text{sens}} \le {\delta _i}, \label{M3_MSE_constraint} \\    
    & \mathbf{A}_i \succeq   0, \ \mathbf{B}_n \succeq 0, \label{M2_constraint3} \\
    & \operatorname{Rank}(\mathbf{A}_i) = 1, \operatorname{Rank}(\mathbf{B}_n) = 1.
    \label{M2_constraint4}
\end{align}
\end{subequations}

Nevertheless, since the transmit and the receive beamforming vectors are coupled, the optimal solution cannot be obtained in polynomial time. Therefore, we decompose optimization problem (\ref{M3}) into two subproblems: the transmit beamforming optimization subproblem and the receive beamforming optimization subproblem. Then, we tackle the overall optimization problem using an alternating optimization approach.

\subsection{Transmit Beamforming Optimization Subproblem}
We first address the subproblem where transmit beamforming vectors ${{\bf{a}}_i}$ and ${{\bf{b}}_n}$ are optimized with fixed receive beamforming vector ${{\bf{v}}_i^{\#}}$. With the MMSE receiver, i.e., fixing receive beamforming vector, the original optimization problem is reformulated as
\begin{equation}
\label{Suboptimization_problem_transmit}
\begin{aligned}[b]
  \mathop{\min}\limits_{\mathbf{A}_i, \mathbf{B}_n} 
    & \sum_{i=1}^I \operatorname{tr}(\mathbf{A}_i) + \sum_{n=1}^N \operatorname{tr}(\mathbf{B}_n) \\
  \text{s.t.}\ 
    & (\ref{M1_SINR_constraint}), (\ref{M3_MSE_constraint})-(\ref{M2_constraint4}).
\end{aligned}
\end{equation}
Notably, the SINR outage probability and rank-one constraints make the problem non-convex.

First of all, we deal with the SINR outage probability constraint (\ref{M1_SINR_constraint}). By introducing the auxiliary variable ${{\bf{d}}_{j,n}^{\#}} = o_{j,n}^{\rm{H}}{\bf{q}}_j^{\text{s}}$, the SINR expression is simplified as equation (\ref{SINR_new}) at the beginning of the following page,
\begin{figure*}[t]
\centering
\begin{equation}
\label{SINR_new}
\Gamma_n
= \dfrac{
    ({\bf{q}}_n^{\text{c}})^{\rm{H}}  \mathbf{B}_n \! \odot \! \left( {\hat{\bf{h}}_n} \hat{\bf{h}}_n^{\rm{H}} \right)^{\rm{T}} \! {\bf{q}}_n^{\text{c}} + \sum_{j=1}^I {\bf{d}}_{j,n}^{\#\rm{H}} \Biggl( R_j^2   {\bf{B}}_n \odot \left( {\hat{p}_{j,n}^{\rm{H}}} {\hat{\bf{g}}_j} {\hat{p}_{j,n}} \hat{\bf{g}}_j^{\rm{H}} \right)
\Biggr) {\bf{d}}_{j,n}^{\#}
}{
    ({\bf{q}}_n^{\text{c}})^{\rm{H}} \! \left( \sum\limits_{\substack{m=1 \\ m \neq n}}^N \!\!\!  \mathbf{B}_m \! \odot \! \left( {\hat{\bf{h}}_n} \hat{\bf{h}}_n^{\rm{H}} \right)^{\rm{T}} \! \right) \! {\bf{q}}_n^{\text{c}} \!+\! 
    ({\bf{q}}_n^{\text{c}})^{\rm{H}} \! \left( \sum\limits_{i=1}^I \!  \mathbf{A}_i \! \odot \! \left( {\hat{\bf{h}}_n} \hat{\bf{h}}_n^{\rm{H}} \right)^{\rm{T}} \! \right) \! {\bf{q}}_n^{\text{c}} \!+\! 
    X_n^{\prime} \!+\! 
    \sigma_n^2
}.
\end{equation}
\vspace{-\baselineskip}
\end{figure*}
where the $X_n^{\prime}$ in (\ref{Xnprime}) is reformulated as 
\begin{equation}
\begin{aligned}[b]
X_n^{\prime} = \sum_{j=1}^I {\bf{d}}_{j,n}^{\#\rm{H}} \Biggl( 
    & R_j^2 \sum_{i=1}^I {\bf{A}}_i \odot \left( \hat{p}_{j,n}^{\rm{H}} {\hat{\bf{g}}_j} {\hat{p}_{j,n}} \hat{\bf{g}}_j^{\rm{H}} \right) \\
    + &R_j^2 \sum_{\substack{m=1 \\ m \neq n}}^N  {\bf{B}}_m \odot \left( \hat{p}_{j,n}^{\rm{H}} {\hat{\bf{g}}_j} {\hat{p}_{j,n}} \hat{\bf{g}}_j^{\rm{H}} \right)
\Biggr) {\bf{d}}_{j,n}^{\#}.
\end{aligned}
\end{equation}
Since the fractional form of SINR in equation (\ref{SINR_new}) is difficult to handle directly when it is constrained by ${\gamma _n}$, we equivalently transform the inequality inside the probability constraint into a polynomial representation, which can be further simplified as (\ref{SINR_inequality}) at the top of the next page.
\begin{figure*}[!t]
\centering
\begin{equation}
\label{SINR_inequality}
\begin{aligned}[b]
&\underbrace{({\bf{q}}_n^{\text{c}})^{\rm{H}} \Biggl( \frac{1}{\gamma_n} \mathbf{B}_n \odot \left( \hat{\bf{h}}_n \hat{\bf{h}}_n^{\rm{H}} \right)^{\rm{T}} 
    - \sum_{\substack{m=1 \\ m \neq n}}^N  \mathbf{B}_m \odot \left( \hat{\bf{h}}_n \hat{\bf{h}}_n^{\rm{H}} \right)^{\rm{T}} 
    - \sum_{i=1}^I  \mathbf{A}_i \odot \left( \hat{\bf{h}}_n \hat{\bf{h}}_n^{\rm{H}} \right)^{\rm{T}} \Biggr) {\bf{q}}_n^{\text{c}}}_{\text{First term}} \\
&+ \underbrace{\sum_{j=1}^I {\bf{d}}_{j,n}^{\#\rm{H}} R_j^2\Biggl( \frac{1}{\gamma_n}\mathbf{B}_n \odot \left( \hat{p}_{j,n}^{\rm{H}} {\hat{\bf{g}}_j} {\hat{p}_{j,n}} \hat{\bf{g}}_j^{\rm{H}} \right)-\sum_{i=1}^I \mathbf{A}_i \odot \left( \hat{p}_{j,n}^{\rm{H}} {\hat{\bf{g}}_j} {\hat{p}_{j,n}} \hat{\bf{g}}_j^{\rm{H}} \right) 
    -  \sum_{\substack{m=1 \\ m \neq n}}^N  \mathbf{B}_m \odot \left( \hat{p}_{j,n}^{\rm{H}} {\hat{\bf{g}}_j} {\hat{p}_{j,n}} \hat{\bf{g}}_j^{\rm{H}} \right) \Biggr) {\bf{d}}_{j,n}^{\#}}_{\text{Second term}} 
    - \sigma_n^2 \geq 0.
\end{aligned}
\end{equation}
\hrulefill
\vspace*{4pt}
\end{figure*}
Then, we decompose the left-hand side of inequality (\ref{SINR_inequality}), without the noise term, into two components, each of which is approximated through Taylor expansion by invoking the following lemma.

\begin{lemma}
Given a complex exponential Gaussian vector
${\boldsymbol{\iota}} = \left( {{e^{j{\theta _1}}}, \ldots ,{e^{j{\theta _K}}}} \right)$, a Gaussian random vector $\boldsymbol{\Lambda} $, and a $K$-order Hermitian matrix $ \bf{Z} $ whose real part ${{\bf{Z}}^R} \in {{\cal S}^K}$ and imaginary part ${{\bf{Z}}^I} \in {{\cal K}^K}$, the second-order Taylor expansion of ${{\boldsymbol{\iota}}^{\rm{H}}}{\bf{Z} }{\boldsymbol{\iota}}$ can be expressed as
\begin{equation}
{{\boldsymbol{\iota}}^{\rm{H}}}{\bf{Z} }{\boldsymbol{\iota}} = \sum\limits_{i,j} {{{\bf{Z}}_{i,j}}}  + {{\boldsymbol{\Lambda }}^{\rm{T}}}{f_1}({{\bf{Z}}^R}){\boldsymbol{\Lambda}} + {{\boldsymbol{\Lambda}}^{\rm{T}}}{f_2}({{\bf{Z}}^I}),
\end{equation}
where the linear maps $f_1$, $f_2$ satisfy
\begin{equation}
[f_1(\mathbf{Z}^R)]_{i,j} = \begin{cases}
\mathbf{Z}^R_{i,j} - \sum_{n=1}^K \mathbf{Z}^R_{i,n}, & \text{if } i = j\\
\mathbf{Z}^R_{i,j}, & \text{if } i \neq j
\end{cases}
\end{equation}
and
\begin{equation}
[f_2(\mathbf{Z}^I)]_i = 2\sum_{n=1}^K \mathbf{Z}^I_{i,n}.
\end{equation}
\end{lemma}
The proof of Lemma 1 can be found in \cite{A_robust_design}.
Next, by setting 
\begin{equation}
\begin{aligned}[b]
{\bf{T}}_n^{\prime} &= \frac{1}{{{\gamma _n}}}{{\bf{B}}_n} \odot {\left( { {{\hat{\bf{h}}_n}} {{{{\hat{\bf{h}}_n}} }^{\rm{H}}}} \right)^{\rm{T}}} \\
&\quad - \sum\limits_{\substack{m=1 \\ m \neq n}}^N {{{\bf{B}}_m} \odot {{\left( {{{\hat{\bf{h}}_n}} {{{{\hat{\bf{h}}_n}} }^{\rm{H}}}} \right)}^{\rm{T}}}} \\
&\quad - \sum\limits_{i = 1}^I { {{\bf{A}}_i} \odot {{\left( {{{\hat{\bf{h}}_n}} {{{{\hat{\bf{h}}_n}} }^{\rm{H}}}} \right)}^{\rm{T}}}},
\end{aligned}
\end{equation}
the first term in inequality (\ref{SINR_inequality}) can be given by
\begin{equation}
({\bf{q}}_n^{\text{c}})^{\rm{H}}{\bf{T}}_n^{\prime}{{\bf{q}}_n^{\text{c}}} = \sum\limits_{i,j} {{\bf{T}}_{n\left[ {i,j} \right]}^{\prime}}  + {{\boldsymbol{\nu }}^{\prime {\rm{T}}}}{\bf{K}}_n^{\prime}{{\boldsymbol{\nu }}^{\prime}} + 2{{\boldsymbol{\nu }}^{\prime {\rm{T}}}}{\bf{U}}_n^{\prime},
\end{equation}
where $\boldsymbol{\nu}$ is a $K$-dimensional standard Gaussian random vector. Besides, $\mathbf{K}_n^{\prime} = ({\sigma_n^{\text{c}}})^{2} ({{\bf{S}}_n^{\text{c}}})^{1/2} f_1\bigl( \mathrm{Re}(\mathbf{T}_n^{\prime}) \bigr) ({{\bf{S}}_n^{\text{c}}})^{1/2
}$, 
and $\mathbf{U}_n^{\prime} = \tfrac{1}{2} {\sigma_n^{\text{c}}} ({{\bf{S}}_n^{\text{c}}})^{1/2} f_2\bigl( \mathrm{Im}(\mathbf{T}_n^{\prime}) \bigr)$.
In the same way, by setting 
\begin{equation}
\begin{aligned}[b]
\mathbf{T}_{j,n}^{\prime\prime} = R_j^2\Biggl(
    &\frac{1}{\gamma_n}\mathbf{B}_n \odot \left( \hat{p}_{j,n}^{\rm{H}} \hat{\mathbf{g}}_j \hat{p}_{j,n} \hat{\mathbf{g}}_j^{\rm{H}} \right) \\
    &- \sum_{i = 1}^I \mathbf{A}_i \odot \left( \hat{p}_{j,n}^{\rm{H}} \hat{\mathbf{g}}_j \hat{p}_{j,n} \hat{\mathbf{g}}_j^{\rm{H}} \right) \\
    &- \sum_{\substack{m = 1 \\ m \neq n}}^N \mathbf{B}_m \odot \left( \hat{p}_{j,n}^{\rm{H}} \hat{\mathbf{g}}_j \hat{p}_{j,n} \hat{\mathbf{g}}_j^{\rm{H}} \right)
\Biggr),
\end{aligned}
\end{equation}
the expression inside the outer summation of the second term in inequality (\ref{SINR_inequality}) can be simplified as
\begin{equation}
{\bf{d}}_{j,n}^{\#\rm{H}}{\bf{T}}_{j,n}^{\prime\prime}{{\bf{d}}_{j,n}^{\#}} = \sum\limits_{i,k} {{\bf{T}}_{j,n\left[ {i,k} \right]}^{\prime\prime}}  + {\boldsymbol{\nu }}_j^{\prime\prime {\rm{T}}}{\bf{K}}_{j,n}^{\prime\prime}{\boldsymbol{\nu }}_j^{\prime\prime} + 2{\boldsymbol{\nu }}_j^{\prime\prime {\rm{T}}}{\bf{U}}_{j,n}^{\prime\prime},
\end{equation}
where $\boldsymbol{\nu}^{\prime\prime}$ is a $K$-dimensional standard Gaussian random vector, $\mathbf{K}_{j,n}^{\prime\prime} = ({\sigma_{j,n}^{\text{st}}})^2 ({{\bf{S}}_{j,n}^{\text{st}}})^{1/2} f_1\bigl( \operatorname{Re}(\mathbf{T}_{j,n}^{\prime\prime}) \bigr) ({{\bf{S}}_{j,n}^{\text{st}}})^{1/2}$, 
and $\mathbf{U}_{j,n}^{\prime\prime} = \tfrac{1}{2} {\sigma_{j,n}^{\text{st}}} ({{\bf{S}}_{j,n}^{\text{st}}})^{1/2} f_2\bigl( \operatorname{Im}(\mathbf{T}_{j,n}^{\prime\prime}) \bigr)$. Herein, ${\sigma_{j,n}^{\text{st}}}$ and ${{\bf{S}}_{j,n}^{\text{st}}}$ are the variance and the normalized covariance matrix of the cascaded channel of ${{\bf{g}}_{i}}$ and ${p_{i,n}}$'s phase error.
Afterwards, we construct $\boldsymbol{\nu }$, ${\bf{K}}_n$ and ${\bf{U}}_n$ as
\begin{equation}
{\boldsymbol{\nu }} = \left({\boldsymbol{\nu }}^{\prime};{\boldsymbol{\nu }}_1^{\prime\prime};\cdots;{\boldsymbol{\nu }}_I^{\prime\prime} \right),
\end{equation}
\begin{equation}
{{\bf{K}}_n} = \operatorname{blkdiag}\left({{\bf{K}}_n^{\prime}},{{\bf{K}}_{1,n}^{\prime\prime}},\cdots,{{\bf{K}}_{I,n}^{\prime\prime}}\right),
\end{equation}
\begin{equation}
{{\bf{U}}_n} = \operatorname{blkdiag}\left( {{\bf{U}}_n^{\prime}}, {{\bf{U}}_{1,n}^{\prime\prime}}, \cdots, {{\bf{U}}_{I,n}^{\prime\prime}} \right).
\end{equation}
Then, the approximate SINR outage probability constraint derived from (\ref{M1_SINR_constraint}) is rewritten as
\begin{equation}
\label{pr_SINR}
\begin{aligned}[b]
& \Pr \biggl\{ 
  \sum_{i,j} \mathbf{T}_{n[i,j]}^{\prime} 
  + \sum_{j=1}^I \sum_{i,k} \mathbf{T}_{j,n[i,k]}^{\prime\prime} 
  \\
& \quad 
  + \boldsymbol{\nu}^{\rm{T}} \mathbf{K}_n \boldsymbol{\nu} 
  + 2 \boldsymbol{\nu}^{\rm{T}} \mathbf{U}_n 
  \le \sigma_n^2 
\biggr\} \le {\varrho_n}.
\end{aligned}
\end{equation}

To address the non-convexity caused by the probabilistic constraint (\ref{pr_SINR}), we convert it into a series of convex constraints by applying the following Lemma 2.

\begin{lemma}
Given ${\bf{Q}} \in {\mathbb{H}^{K \times K}}$, ${\bf{r}} \in {\mathbb{R}^{K \times K}}$ and ${\bf{e}}\sim{\cal N}({\bf{0}},{\bf{I}})$, for any $\tau > 0$ and $\mu > 1/\sqrt{2}$, it holds that
\begin{equation}
\label{lemma2inequality}
\begin{aligned}[b]
&\Pr \left\{ {{{\bf{e}}^{\rm{T}}}{\bf{Qe}} + 2{\mathop{\rm Re}\nolimits} \left\{ {{{\bf{e}}^{\rm{T}}}{\bf{r}}} \right\} + s \le 0} \right\}\\
&\le \left\{ {\begin{array}{*{20}{l}}
{\exp \left( { - \frac{{{\tau ^2}}}{{4{\Upsilon ^2}}}} \right),}&{0 < \tau  \le 2\lambda \mu \Upsilon }\\
{\exp \left( { - \frac{{\tau \lambda \mu }}{\Upsilon } + {{(\lambda \mu )}^2}} \right),}&{\tau  > 2\lambda \mu \Upsilon, }
\end{array}} \right.
\end{aligned}
\end{equation}
\end{lemma}
where $\lambda  = 1 - \left( {1/(2{\mu ^2})} \right)$, $s = \tau  - {\mathop{\operatorname{tr}}\nolimits} ({\bf{Q}})$, and $\Upsilon  = \mu {\left\| {\bf{Q}} \right\|_F} + (1/\sqrt 2 )\left\| {\bf{r}} \right\|$. 
Lemma 2's proof is available in \cite{Outage_constrained}.
By appropriately selecting $\tau$, we can make the inequality (\ref{lemma2inequality}) in Lemma 2 match the form of the SINR outage probability constraint (\ref{pr_SINR}). Specifically, we set the right-hand side expressions of (\ref{lemma2inequality}) in Lemma 2 equal to ${\varrho_n}$ and then we get ${\tau _1}$ and ${\tau _2}$ as
\begin{equation}
\begin{array}{l}
{\tau _1} = 2\sqrt {\ln \frac{1}{{{\varrho_n}}}} \Upsilon , \
{\tau _2} = \left( {\lambda \mu  + \frac{{\ln \frac{1}{{{\varrho_n}}}}}{{\lambda \mu }}} \right)\Upsilon .
\end{array}
\end{equation}
Since the right-hand side of (\ref{lemma2inequality}) decreases monotonically with $\tau$, the search for the minimal $\tau$ is required. We choose $\mu  > \frac{1}{{\sqrt 2 }}$ while making it satisfy $\lambda \mu = \sqrt{ \ln (1 / \varrho_n) }$ so as to equate the minimum value of ${\tau _1}$ and ${\tau _2}$. Additionally, considering the definition of $\lambda$, the parameter $\mu$ can be obtained. In this way, we can transform the SINR outage probability constraint (\ref{M1_SINR_constraint}) as
 \begin{equation}
 \label{SINR_tranformed}
    \tau  \ge 2\sqrt {\ln \frac{1}{{{\varrho_n}}}} \Upsilon.
 \end{equation}
Furthermore, by combining with the definition of $\tau$ and $\Upsilon$, we transform constraint (\ref{SINR_tranformed}) into three convex constraints:
\begin{align}
\operatorname{tr} \left( \mathbf{K}_n \right) + \mathbf{s}_n &\geq 2\sqrt{\ln \frac{1}{\varrho_n}} \left( x_n + y_n \right), \label{SINRtrans_1} \\
\frac{1}{\sqrt{2}} \left\| \mathbf{U}_n \right\| &\leq x_n, \label{SINRtrans_2}\\
\mu_n \left\| \mathbf{K}_n \right\|_F &\leq y_n, \label{SINRtrans_3}
\end{align}
where $x_n$ and $y_n$ are two auxiliary variables. Besides, ${{\bf{s}}_n} = \sum\limits_{i,j} {{\bf{T}}_{n\left[ {i,j} \right]}^{\prime}}  + \sum\limits_{j = 1}^I {\sum\limits_{i,k} {{\bf{T}}_{j,n\left[ {i,k} \right]}^{\prime\prime}} } - \sigma _n^2$.
Based on this, we can reformulate problem (\ref{Suboptimization_problem_transmit}) as 
\begin{subequations}
\label{Suboptimization_problem_change_SINR}
\begin{align}
  \mathop{\min}\limits_{\mathbf{A}_i, \mathbf{B}_n, x_n, y_n} 
    & \sum_{i=1}^I \operatorname{tr}(\mathbf{A}_i) + \sum_{n=1}^N \operatorname{tr}(\mathbf{B}_n) \nonumber \\
  \text{s.t.}\ 
    & (\ref{M3_MSE_constraint})-(\ref{M2_constraint4}),(\ref{SINRtrans_1}),(\ref{SINRtrans_2}),(\ref{SINRtrans_3}). \tag{53}
\end{align}
\end{subequations}

The non-convexity of Problem (\ref{Suboptimization_problem_change_SINR}) is now solely attributed to the rank-one constraints (\ref{M2_constraint4}). To resolve this, a penalty function is incorporated into the objective function. Since both ${{\bf{A}}_i}$ and ${{\bf{B}}_n}$ are positive semidefinite matrices with nonnegative eigenvalues, the rank-one constraints imply that each matrix has a single nonzero eigenvalue, while all others are zero. The transformed rank-one constraints can be expressed as
\begin{equation}
\begin{array}{l}
{\mathop{\operatorname{tr}}\nolimits} \left( {{{\bf{A}}_i}} \right) - {\lambda _{i,\max }} = 0,\ {\mathop{\operatorname{tr}}\nolimits} \left( {{{\bf{B}}_n}} \right) - \lambda _{n,\max }^{\prime} = 0,
\end{array}
\end{equation}
where ${\lambda _{i,\max }}$ and $\lambda _{n,\max }^{\prime}$ are ${{\bf{A}}_i}$'s and ${{\bf{B}}_n}$'s maximum eigenvalues.
Thus, the new objective function with penalty terms can be formulated as 
\begin{equation}
\label{penalty_function}
\begin{aligned}[b]
\mathop {\min }\limits_{{{\bf{A}}_i},{{\bf{B}}_n},{x_n},{y_n}} &\sum\limits_{i = 1}^I {{\mathop{\operatorname{tr}}\nolimits} \left( {{{\bf{A}}_i}} \right)}  + \sum\limits_{n = 1}^N {{\mathop{\operatorname{tr}}\nolimits} \left( {{{\bf{B}}_n}} \right)}\\  
&+ {\rho _1}\sum\limits_{i = 1}^I {\left( {{\mathop{\operatorname{tr}}\nolimits} \left( {{{\bf{A}}_i}} \right) - {\lambda _{i,\max }}} \right)}\\  
&+ {\rho _2}\sum\limits_{n = 1}^N {\left( {{\mathop{\operatorname{tr}}\nolimits} \left( {{{\bf{B}}_n}} \right) - \lambda _{n,\max }^{\prime}} \right)},
\end{aligned}
\end{equation}
where ${\rho _1}$ and ${\rho _2}$ denote the penalty factors. Notably, the penalty factors ${\rho _1}$ and ${\rho _2}$ critically affect the subproblem's convergence speed and solution accuracy. Therefore, it is essential to select appropriate initial values for ${\rho _1}$ and ${\rho _2}$. 

Notice that the resulting objective function (\ref{penalty_function}) remains non-convex because of the presence of these penalty terms. Therefore, an iterative method is employed to overcome this difficulty. For the $t$-th iteration results ${\bf{A}}_i^{(t)}$ and ${\bf{B}}_n^{(t)}$, the following holds:
\begin{equation}
\begin{aligned}[b]
&{{\mathop{\operatorname{tr}}\nolimits}  \left( {{\bf{A}}_i^{(t + 1)}} \right) - {{\left( {{\boldsymbol{\zeta }}_{i,\max }^{(t)}} \right)}^{\rm{H}}}{\bf{A}}_i^{(t + 1)}{\boldsymbol{\zeta }}_{i,\max }^{(t)}}\\
&{\quad  \ge {\mathop{\operatorname{tr}}\nolimits} \left( {{\bf{A}}_i^{(t + 1)}} \right) - \lambda _{i,\max }^{(t + 1)} \ge 0},
\end{aligned}
\end{equation}
\begin{equation}
\begin{aligned}[b]
&{{\mathop{\operatorname{tr}}\nolimits} \left( {{\bf{B}}_n^{(t + 1)}} \right) - {{\left( {{\boldsymbol{\zeta }}_{n,\max }^{\prime(t)}} \right)}^{\rm{H}}}{\bf{B}}_n^{(t + 1)}{\boldsymbol{\zeta }}_{n,\max }^{\prime(t)}}\\
&{\quad  \ge {\mathop{\operatorname{tr}}\nolimits} \left( {{\bf{B}}_n^{(t + 1)}} \right) - \lambda _{n,\max }^{\prime(t + 1)} \ge 0},
\end{aligned}
\end{equation}
where ${{\boldsymbol{\zeta }}_{i,\max }}$ and ${\boldsymbol{\zeta }}_{n,\max }^{\prime}$ are the unit eigenvectors corresponding to ${\lambda _{i,\max }}$ and $\lambda _{n,\max }^{\prime}$, respectively. Based on this, the subproblem (\ref{Suboptimization_problem_transmit}) can finally be reformulated as 
\begin{subequations}
\label{Suboptimization_problem_change_Rank1}
\begin{align}
  \mathop{\min}\limits_{\mathbf{A}_i, \mathbf{B}_n, x_n, y_n} 
    & \sum_{i=1}^I \operatorname{tr}\left( \mathbf{A}_i^{(t+1)} \right) + \sum_{n=1}^N \operatorname{tr}\left( \mathbf{B}_n^{(t+1)} \right) \nonumber\\
    & {+} \rho_1 \sum_{i=1}^I \biggl( \operatorname{tr}\left( \mathbf{A}_i^{(t+1)} \right) {-} \left(\! \boldsymbol{\zeta }_{i,\max}^{(t)} \!\right)^{\!\rm{H}} \mathbf{A}_i^{(t+1)} \boldsymbol{\zeta }_{i,\max}^{(t)} \biggl) \nonumber\\
    & {+} \rho_2 \sum_{n=1}^N \biggl( \operatorname{tr}\left( \mathbf{B}_n^{(t+1)} \right) {-} \left(\! \boldsymbol{\zeta }_{n,\max}^{\prime(t)} \!\right)^{\!\rm{H}} \mathbf{B}_n^{(t+1)} \boldsymbol{\zeta }_{n,\max}^{\prime(t)} \biggl) \nonumber\\
  \text{s.t.}\ 
    & (\ref{M3_MSE_constraint}),(\ref{M2_constraint3}), (\ref{SINRtrans_1}),(\ref{SINRtrans_2}),(\ref{SINRtrans_3}). \tag{58}
\end{align}
\end{subequations}

Then, the subproblem is transformed into a semidefinite programming (SDP) problem, solvable via convex optimization toolboxes like CVX \cite{CVX}.

 \subsection{Receive Beamforming Optimization Subproblem}
Next, we consider the subproblem where the receive beamforming vector ${\bf{v}}_i$ is optimized with fixed transmit beamforming vectors ${{\bf{a}}_i}$ and ${{\bf{b}}_n}$. In this situation, the objective function is fixed. Since the receive beamforming vector is used only for receiving reflected signals from sensing targets at the satellite, only MSE constraint is related to it. Therefore, we need to find the receive beamforming vector ${\bf{v}}_i$ that minimizes the MSE. We tackle this subproblem with an optimized approach. Specifically, the value of MSE is taken as the objective function, so the optimization problem for ${{\bf{v}}_i^{\#}}$ can be given by
\begin{equation} 
\label{MV}
\begin{aligned}[b]
\mathop {\min }\limits_{{{\bf{v}}_i^{\#}}} 
     \widetilde{\operatorname{MSE}}_i^{\text{sens}}.
\end{aligned}
\end{equation}

To convexify the problem, an auxiliary variables ${{\bf{V}}_i^{\#}} = {{\bf{v}}_i^{\#}}({{\bf{v}}_i^{\#}})^{\rm{H}}$ is introduced and used to reformulate the objective function (\ref{MV}). In addition, a penalty term $\rho_3 \sum_{i=1}^I \biggl( \operatorname{tr}\left( {{\bf{V}}_i^{\# (t+1)}} \right) {-} \left(\! \boldsymbol{\zeta }_{i,\max}^{\prime\prime(t)} \!\right)^{\!\rm{H}} {{\bf{V}}_i^{\# (t+1)}} \boldsymbol{\zeta }_{i,\max}^{(t)} \biggl)$ is incorporated into the objective function to deal with the rank-one constraint, where $\rho_3$ and $\boldsymbol{\zeta }_{i,\max}^{\prime\prime(t)}$ denote the penalty factor and the unit eigenvector associated with the maximum eigenvalue $\lambda _{i,\max }^{\prime\prime(t)}$ of ${{\bf{V}}_i^{\#(t)}}$. Therefore, suboptimization problem (\ref{MV}) can be converted to a convex problem (\ref{optimizationV}) at the top of the following page.
\begin{figure*}[!t]
\centering
\begin{equation}
\label{optimizationV}
\begin{aligned}[b]
\mathop {\min }\limits_{{{\bf{V}}_i^{\#}}} 
     &\left| \operatorname{tr}\left({{\bf{V}}_i^{\#}}\operatorname{diag}\left( {\hat{\bf{g}}_i} \right){\bf{Q}}_i^{\text{s}}\operatorname{diag}\left( {\hat{\bf{g}}_i}^{\rm{H}} \right)
    \biggl( \sum_{l=1}^{I} {\bf{A}}_l + \sum_{n=1}^{N} {\bf{B}}_n \biggr)
    \operatorname{diag}\left( {\hat{\bf{g}}_i} \right){\bf{Q}}_i^{\text{s}}\operatorname{diag}\left( {\hat{\bf{g}}_i}^{\rm{H}} \right)\right) - 1 \right| R_i^2 \\
&+ \sum_{\substack{j=1 \\ j \ne i}}^{I} \operatorname{tr}\left({{\bf{V}}_i^{\#}}\operatorname{diag}\left( {\hat{\bf{g}}_j} \right){\bf{Q}}_j^{\text{s}}\operatorname{diag}\left( {\hat{\bf{g}}_j}^{\rm{H}} \right)
    \biggl(\sum_{l=1}^{I} {\bf{A}}_l + \sum_{n=1}^{N} {\bf{B}}_n \biggr)
    \operatorname{diag}\left( {\hat{\bf{g}}_j} \right){\bf{Q}}_j^{\text{s}}\operatorname{diag}\left( {\hat{\bf{g}}_j}^{\rm{H}} \right)\right) R_j^2 \\
    &+ (\sigma^{\prime})^2 \operatorname{tr}({{\bf{V}}_i^{\#}})+\rho_3 \sum_{i=1}^I \biggl( \operatorname{tr}\left( {{\bf{V}}_i^{\# (t+1)}} \right) {-} \left(\! \boldsymbol{\zeta }_{i,\max}^{\prime\prime(t)} \!\right)^{\!\rm{H}} {{\bf{V}}_i^{\# (t+1)}} \boldsymbol{\zeta }_{i,\max}^{\prime\prime(t)} \biggl).
\end{aligned}
\end{equation}
\hrulefill
\vspace*{4pt}
\end{figure*}

Finally, the original optimization problem is converted to multiple SDP problems, solvable via convex optimization toolboxes such as CVX. By performing eigenvalue decomposition (EVD), the optimal solution to problem (\ref{M1}) is derived, i.e.,
\begin{equation}
    \label{EVD}
    \begin{aligned}
    {{\bf{a}}_i^{*}}&=\sqrt {{\lambda _{i,\max }}\left( {{\bf{A}}_i^{*}} \right)} {{\boldsymbol{\zeta }}_{i,\max }^{*}}, \\
    {{\bf{b}}_n^{*}}&=\sqrt {\lambda _{n,\max }^{\prime}\left( {{\bf{B}}_n^{*}} \right)} {\boldsymbol{\zeta }}_{n,\max }^{\prime*},\\
    {{\bf{v}}_i^{\#*}}&=\sqrt {\lambda _{i,\max }^{\prime\prime}\left( {{\bf{V}}_i^{\#*}} \right)} \boldsymbol{\zeta }_{i,\max}^{\prime\prime*},
    \end{aligned}
\end{equation}
where ${\lambda _{i,\max }}\left( {{\bf{A}}_i^{*}} \right)$, $\lambda _{n,\max }^{\prime}\left( {{\bf{B}}_n^{*}} \right)$ and $\lambda _{i,\max }^{\prime\prime}\left( {{\bf{V}}_i^{\#*}} \right)$ denote the maximum eigenvalues of optimal solutions ${{\bf{A}}_i^{*}}$, ${{\bf{B}}_n^{*}}$ and ${{\bf{V}}_i^{\#*}}$, respectively. In addition, ${\boldsymbol{\zeta }}_{i,\max }^{*}$, ${\boldsymbol{\zeta }}_{n,\max }^{\prime*}$ and $\boldsymbol{\zeta }_{i,\max}^{\prime\prime*}$ represent the corresponding unit eigenvectors. As for the penalty factors $\rho_1$, $\rho_2$, and $\rho_3$, they are initialized to balance the magnitude of the penalty terms with that of the original objective function. Specifically, in the subsequent simulations, $\rho_1$ and $\rho_2$ are initially set to $0.5$ to match the total transmit power, while the initial value of $\rho_3$ is determined according to the maximum tolerance of MSE to maintain a comparable scale. During the alternating iterative procedure, each penalty factor is multiplied by a factor $\kappa = 2$ at the end of every iteration, thereby progressively enforcing the rank-one constraints. The selection of $\kappa$ strikes a balance between computational stability and convergence rate, as an excessively large value may lead to computing difficulties while a smaller one results in slow convergence. It is worth noting that too extreme values of the initial penalty factors can lead to unstable convergence behavior, such as fluctuating around the limit point. Combining the above steps, the robust beamforming design for ISAC in LEO satellite systems is summarized in Algorithm 1.

\begin{algorithm}[!ht]
    \renewcommand{\algorithmicrequire}{\textbf{Input:}}
	\renewcommand{\algorithmicensure}{\textbf{Output:}}
	\caption{Robust Beamforming Design for ISAC in LEO Satellite Systems}
    \label{Robust_Algorithm}
    \begin{algorithmic}[1] 
        \REQUIRE  K, I, N, ${\gamma _n}$, ${\varrho_n}$, ${\mu _n}$, ${\delta _i}$, $\sigma _n^2$; 
	    \ENSURE ${{\bf{a}}_i}$, ${{\bf{b}}_n}$, ${\bf{v}}_i^{\#}$; 
        
        \STATE \textbf{Initialize} iteration index $t = 1$, ${\bf{a}}_i^{(0)}$, ${\bf{b}}_n^{(0)}$, $\forall i$, n, ${x_n} = {y_n} = 0,\ \forall n$;
        \STATE \textbf{Set} maximal iteration number $L_{\max}$ and ${M_{{\max}}}$, accuracy $\varepsilon$, coefficient $\kappa$ and penalty factor ${\rho _1}$, ${\rho _2}$.
        \STATE Obtain $\left\{ {{\bf{A}}_i^{(0)},{\bf{B}}_n^{(0)}} \right\}$;
        \REPEAT 
            \REPEAT
            \STATE Obtain ${\bf{V}}_i^{\#(t,l)}$ by solving (\ref{optimizationV}) with ${\bf{A}}_i^{(t-1)}$, ${\bf{B}}_n^{(t-1)}$;
            \IF {$\left| {{\mathop{\operatorname{tr}}\nolimits} \left( {{\bf{V}}_i^{\#(t,l)}} \right) - \lambda _{i,\max }^{\prime\prime(t,l)}} \right| > \varepsilon$ }
                 \STATE Update the penalty factor $\rho _3^{(t,l)} = \kappa \rho _3^{(t,l + 1)}$;
            \ENDIF
            \STATE Update iteration number $l = l + 1$;
            \UNTIL {$l = {L_{\max}}$ or solution converges}
            \STATE Use EVD to obtain ${\bf{v}}_i^{\#(t)}$
            \REPEAT
                \STATE Calculate ${\bf{K}}_n^{(t,m)}$, ${\bf{U}}_n^{(t,m)}$;
                \STATE Obtain ${\bf{A}}_i^{(t,m)}$, ${\bf{B}}_n^{(t,m)}$ by solving problem (\ref{Suboptimization_problem_change_Rank1});
                \IF {$\left| {{\mathop{\operatorname{tr}}\nolimits} \left( {{\bf{A}}_i^{(t,m)}} \right) - \lambda _{i,\max }^{(t,m)}} \right| > \varepsilon$ }
                    \STATE Update the penalty factor $\rho _1^{(t,m)} = \kappa \rho _1^{(t,m + 1)}$;
                \ENDIF
                \IF {$\left| {{\mathop{\operatorname{tr}}\nolimits} \left( {{\bf{B}}_n^{(t,m)}} \right) - \lambda _{n,\max }^{\prime(t,m)}} \right| > \varepsilon$ }
                    \STATE Update the penalty factor $\rho _2^{(t,m)} = \kappa \rho _2^{(t,m + 1)}$;
                \ENDIF 
                \STATE Update iteration number $m = m + 1$;
            \UNTIL {$m = {M_{\max}}$ or solution converges}
            \STATE Obtain $\left\{ {{\bf{A}}_i^{(t)},{\bf{B}}_n^{(t)}} \right\}$;
            \STATE Update iteration index $t = t + 1$;
        \UNTIL {convergence.}
        \STATE Obtain ${\bf{a}}_i^*$ and ${\bf{b}}_n^*$ by performing EVD on ${\bf{A}}_i^*$ and ${\bf{B}}_n^*$ according to (\ref{EVD});
    \end{algorithmic}
\end{algorithm}

\subsection{Algorithm Analysis}
In this section, the proposed algorithm's convergence and complexity is analyzed.

\textit{Convergence Analysis:}
Given the MMSE receiver, i.e., fixing the receive beamforming vector ${\bf{v}}_i$, the subproblem (\ref{Suboptimization_problem_change_Rank1}) is convex over optimization variables $\mathbf{A}_i$, $\mathbf{B}_n$, $x_n$ and $y_n$, and thus can be solved by CVX, which ensures that the result of each iteration will be less than the last one. Meanwhile, subproblem (\ref{optimizationV}) is also convex which ensures that the result ${\bf{v}}_i$ of the new iteration can reduce the MSE under the same value of $\mathbf{A}_i$ and $\mathbf{B}_n$. Thus, the total transmit power decreases monotonically. Besides, QoS constraints in (\ref{M1_MSE_constraint})-(\ref{M1_SINR_constraint}) ensure a lower bound on the total transmit power. Therefore, Algorithm 1 converges based on the monotone bounded convergence (MBC) theorem \cite{Real_Analysis}. Moreover, as shown in Fig. \ref{fig:cvg}, the algorithm is able to converge quickly within several iterations.

\begin{figure}
    \centering
    \includegraphics[width=1\linewidth]{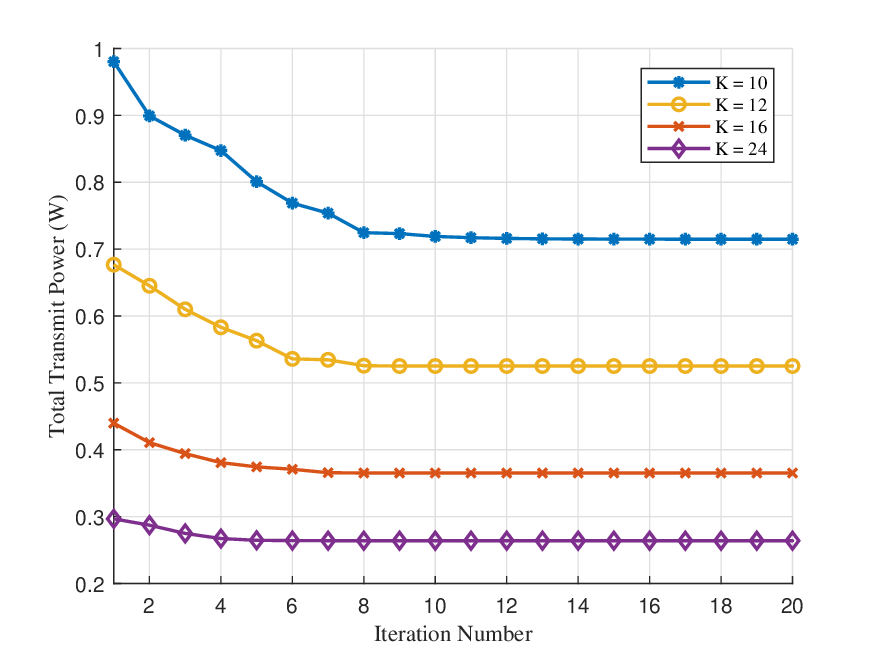}
    \caption{Convergence behavior of Algorithm 1.}
    \label{fig:cvg}
\end{figure}

\textit{Complexity Analysis:}
Considering that the proposed algorithm is iterative and performs the same steps in each iteration, our analysis is therefore devoted to the complexity of each iteration. As can be observed, the main computational complexity comes from (\ref{Suboptimization_problem_change_Rank1}). Containing only second order cone (SOC) and linear matrix inequality (LMI) constraints, the optimization problem is solvable via a standard inter-point method (IPM) \cite{Convex_Optimization}. Iteration computation cost and iteration complexity constitute the two primary components of a generic IPM's complexity. With a certain $\varsigma \ {\rm{ > }}\ 0$, the iteration complexity of solving an $\varsigma$-optimal solution is on the order of $ \ln \left( {{1 \mathord{\left/{\vphantom {1 \varsigma }} \right.\kern-\nulldelimiterspace} \varsigma }} \right)\psi$, with $\psi$ being the barrier parameter which measures the conic constrains' geometic complexity \cite{Outage_constrained}. Meanwhile, the coefficient matrix's construction and factorization determines the iteration computation cost. Considering that problem (\ref{Suboptimization_problem_change_Rank1}) has $N+I$ LMI constraints of dimension 1, $N+I$ LMI constraints of size K and $N$ SOC constraints of dimension $\left( {I + 1} \right)K + 1$, $N$ SOC constrains of size $\left( {I + 1} \right)K^2 + 1$ and the decision variable $n=\mathcal{O}\left( (I + N)K^2 \right)$. Therefore, the worst case complexity for per iteration of Algorithm 1 is on the order of $\ln \left( {{1 \mathord{\left/
 {\vphantom {1 \varsigma }} \right. \kern-\nulldelimiterspace} \varsigma }} \right)\psi $. Herein $\begin{aligned}[t]
\psi = &\sqrt{(N+I)(K+1)+4N} \cdot n \\
&\cdot \bigl[ (N+I)(K^3+1) + n((N+I)(K^2+1)) \\
&\quad + N\bigl((K(I+1)+1)^2 + (K^2(I+1)+1)^2\bigr) + n^2 \bigr],
\end{aligned}$
with decision variable $n=\mathcal{O}\left( (I + N)K^2 \right)$.

\section{Simulation}
In this section, the effectiveness of the proposed algorithm for ISAC in LEO satellite systems is evaluated through extensive simulations. The simulation parameters, unless specified, are as shown in Table I.

\begin{table*}[htb]   
\normalsize
\begin{center}   
\caption{Simulation Parameters}  
\label{table:1} 
\begin{tabular}{|c|c|}   
\hline   \textbf{Parameter} & \textbf{Value}  \\
\hline   Satellite orbit & LEO  \\
\hline   Number of satellite antennas & $K=10$   \\ 
\hline   Number of sensing targets & $I=2$   \\  
\hline   Number of CUs & $N=8$   \\      
\hline   Required minimum SINR & $\gamma_i=3\ \mathrm{dB}$   \\
\hline   Maximum tolerance of MSE & $\delta_n=0.1$   \\
\hline   SINR outage probability & $p_n=0.1$   \\
\hline   RMS of target reflection coefficient & $R_i=1$   \\
\hline   Standard deviation of phase error & $\sigma _n^{\text{c}}=\sigma _i^{\text{s}}=\sigma _{i,n}^{\text{t}}=\sigma_0=0.1$   \\
\hline   Normalized covariance matrices of phase error & ${\bf{S}}_n^{\text{c}}={\bf{S}}_i^{\text{s}}=\bf{I}$   \\
\hline   Carrier frequency & $f_c=5\ \mathrm{GHz}$   \\
\hline   Bandwidth & $B=25\ \mathrm{MHz}$   \\
\hline   Boltzmann constant & $\kappa=1.38 \times {10^{ - 23}}\ \mathrm{J/m}$ \\
\hline   Noise temperature & $T=300\ \mathrm{K}$   \\
\hline   Distance between satellite and CUs & $d_n\in [1000 \sim 2000]\ \mathrm{km}$   \\
\hline   Distance between satellite and sensing targets & $d_0^{\prime} =1000\ \mathrm{km}$   \\
\hline   Distance between sensing targets and CUs & $d_{i,n}^{\prime\prime}\in [400 \sim 600]\ \mathrm{km}$   \\
\hline   CU receiving antenna gain & $G_n=3\ \mathrm{dBi}$   \\
\hline   Maximum satellite antenna gain & $M=53\ \mathrm{dBi}$   \\
\hline   3-dB angle & ${\varphi ^{3dB}}=0.4^{\circ}$   \\
\hline   Rain fading mean & ${\mu _{{{\rho}}_n^{\text{c}}}}={\mu _{{{\rho}}_i^{\text{s}}}}={\mu _{{{\rho}}_{i,n}^{t}}}=-1\ \mathrm{dB}$\\
\hline   Rain fading variance & ${\mkern 1mu} \sigma _{{{\rho}}_n^{\text{c}}}^2={\mkern 1mu} \sigma _{{{\rho}}_i^{\text{s}}}^2={\mkern 1mu} \sigma _{{{\rho}}_{i,n}^{\text{t}}}^2=0.5\ \mathrm{dB}$  \\
\hline   Rician factor & ${{\lambda^{\text{c}} _n=\lambda _i^{\text{s}}}}=5$   \\
\hline
\end{tabular}   
\end{center}   
\end{table*}

Firstly, in Fig. \ref{fig:cvg}, the proposed algorithm's convergence performance is evaluated under different satellite antenna numbers. We can observe that in all cases the transmit power decreases monotonically with iterations and converges within only a few iterations. Hence, the proposed algorithm can be applied over fast time-varying LEO satellite channels. Moreover, the transmit power declines as more satellite antennas are added, indicating that the algorithm performance can be enhanced by deploying more satellite antennas.

\begin{figure}
    \centering
    \includegraphics[width=1\linewidth]{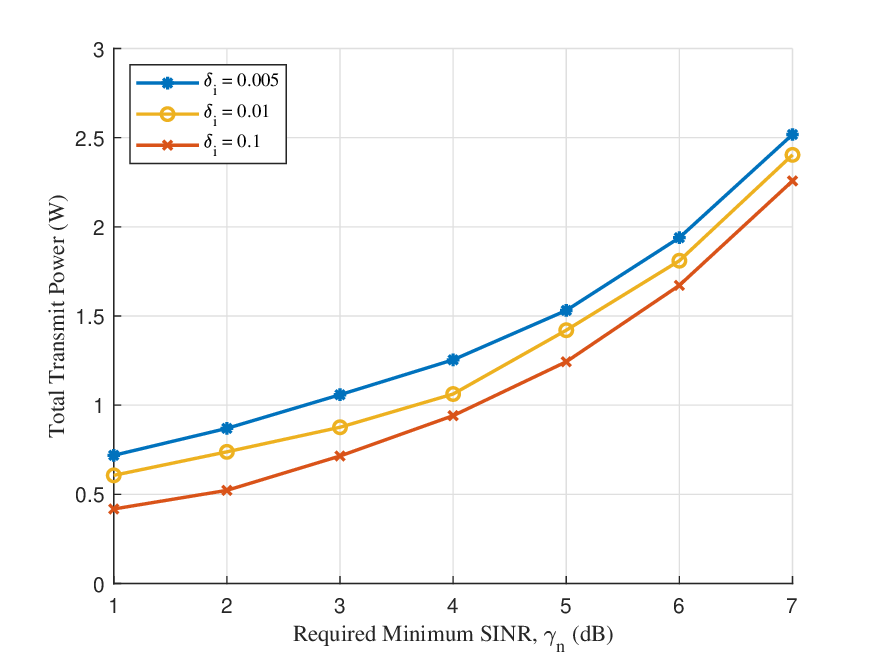}
    \caption{Total transmit power versus the required minimum SINR for different maximum tolerance of MSE.}
    \label{fig:diffMSE}
\end{figure}

Then, we investigate the required minimum SINR of communication ${\gamma _n}$ and the maximum tolerance of MSE ${\delta _i}$'s impact on transmit power. As shown in Fig. \ref{fig:diffMSE}, a rise in ${\gamma _n}$ drives an accelerating increase in transmit power. This indicates that a larger $\gamma_n$ represents a stricter constraint on SINR. With the noise fixed, the LEO satellite requires more transmit power, i.e., increasing the power of the beams to mitigate the influence of noise and interference. In terms of the change in the maximum tolerance of MSE, we can observe that the larger ${\delta _i}$ is, the lower the power, since larger maximum tolerance of MSE represents lower sensing QoS requirements \footnote{ The selection of the maximum tolerance of MSE $\delta_i$ is grounded in the RMS value of the target reflection coefficient $R_i$. Given that $R_i = 1$ in our setup, the investigated thresholds $\delta_i \in \{0.1, 0.01, 0.005\}$ effectively represent different levels of sensing precision, corresponding to normalized MSE levels of $10\%$, $1\%$, and $0.5\%$ with respect to $|R_i|^2$, respectively. Notably, $\delta_i = 0.1$ is designated as the default value to ensure problem feasibility across the investigated phase uncertainty range.}.  By relaxing this sensing requirement, the LEO satellite can reduce the power of sensing beams. This simultaneously reduces the interference beams reflected from the sensing targets to the CUs, making it easier to meet communication QoS requirements. Therefore, balancing power consumption and sensing effectiveness is of great significance.

\begin{figure}
    \centering
    \includegraphics[width=1\linewidth]{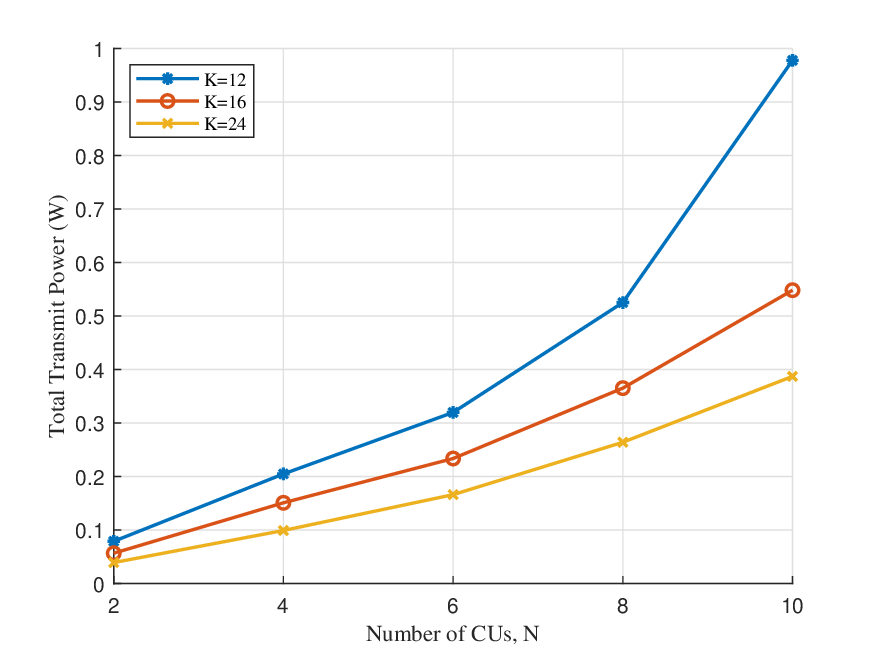}
    \caption{Total transmit power versus the number of CUs with different satellite antennas number.}
    \label{fig:diffN}
\end{figure}

Next, we investigate the influence of satellite antenna number $K$ and CU count $N$ on the proposed algorithm's performance. As shown in Fig. \ref{fig:diffN}, with a fixed $K$, the total transmit power increases at an ever faster rate as more CUs are served. This is because as the number of CUs rises, more beams are needed to transmit communication signals, meaning more power with the same QoS requirements. In addition, interference from both satellite-transmitted and target-reflected non-corresponding communication signals has increased, which leads to the rise in transmit power. Furthermore, from the relationship of the three curves in the above figure, it can be observed that the curve with a larger $K$ always has a lower total transmit power than the curve with a smaller $K$. That is, adding more satellite antennas serves to enhance the algorithm's performance, as more antennas provide higher spatial multiplexing gain, thereby enhancing the performance. However, the improvement in algorithm performance brought about by increasing the number of satellite antennas is not unlimited. As we can see, the performance gap between antenna number $K = 10$ and $K = 12$ is less than that between antenna number $K = 8$ and $K = 10$. Since more satellite antennas leads to increased overhead, balancing cost and performance becomes necessary.

\begin{figure}
    \centering
    \includegraphics[width=1\linewidth]{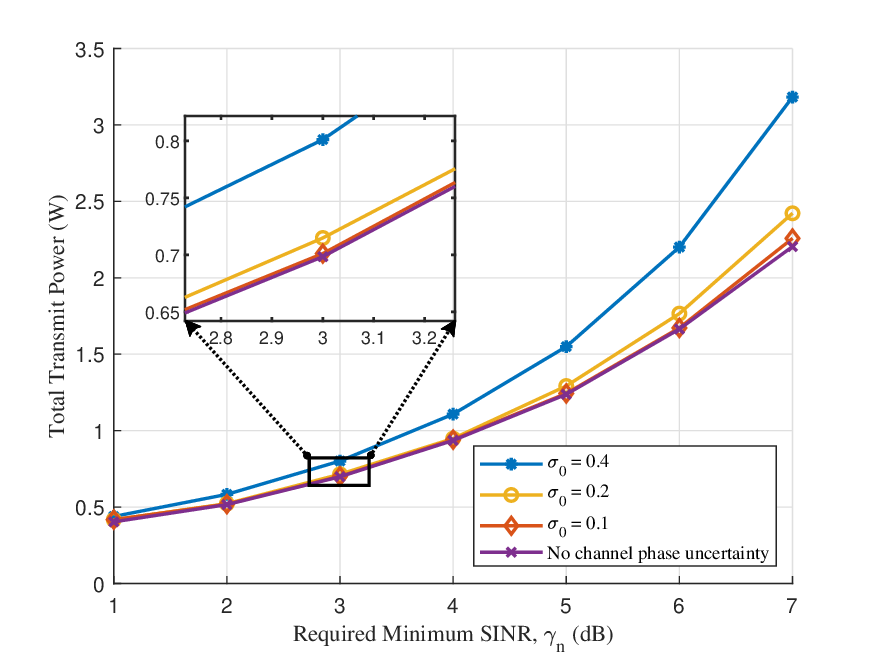}
    \caption{Total transmit power versus required minimum SINR under different channel phase uncertainties.}
    \label{fig:phase_uncertainty}
\end{figure}

Further, we study the influence of the variance of phase error on the total transmit power. Since phase uncertainty exists in satellite channels, the imperfect CSI should be taken into account. Notably, the standard deviation $\sigma_0$ here also effectively captures the residual Doppler-induced phase rotations stemming from target mobility. The investigated range of $\sigma_0$ covers various channel conditions, ranging from relatively small impairments ($\sigma_0 = 0.1$) to highly dynamic environments ($\sigma_0 = 0.4$). Fig. \ref{fig:phase_uncertainty} shows that the transmit power rises with rising phase uncertainty, especially when higher SINR requirement is needed. It is because beamforming requires precise phase control to form direct beams, which is distorted by the channel phase error. Since higher SINR requirement means greater beamforming demands, more transmit power is needed to meet QoS requirements. Notably, the gap between the case of no channel phase uncertainty case and the case of $\sigma_0=0.1$ is minor, which proves that the algorithm has good robustness. However, as $\sigma_0$ further increases, the required power rises at an accelerating rate, indicating that excessive cumulative errors—stemming from both CSI acquisition (e.g., noise and quantization) and residual target mobility—will eventually lead to more significant performance degradation.

\begin{figure}
    \centering
    \includegraphics[width=1\linewidth]{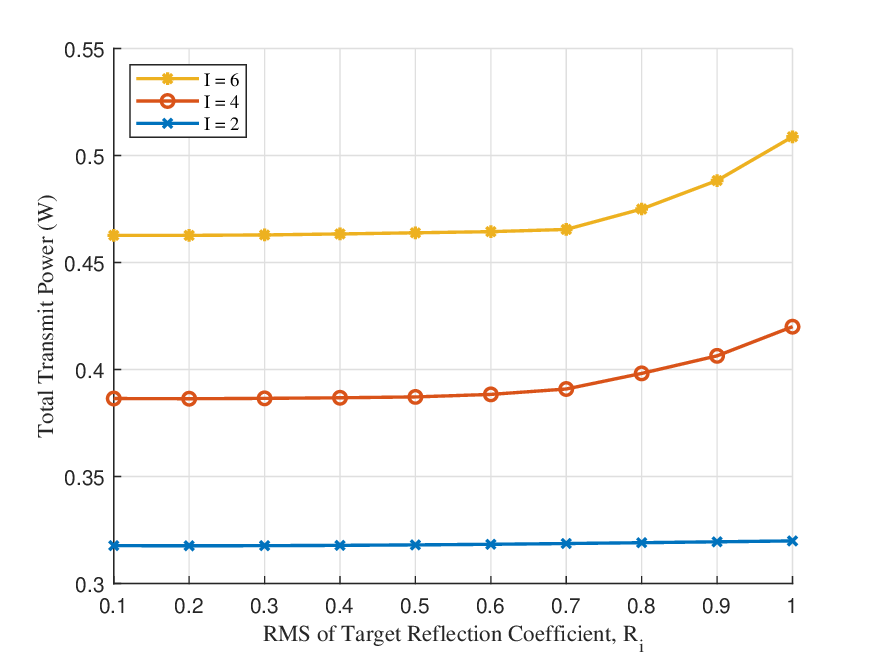}
    \caption{Total transmit power versus RMS of target reflection coefficient with different number of sensing targets.}
    \label{fig:diff_Ri}
\end{figure}

In addition, since the proposed algorithm estimates the reflection coefficient under the premise that the RMS value of the sensing target's reflection coefficient is known, the impact of $R_i$ on the algorithm's performance needs to be examined. In Fig. \ref{fig:diff_Ri} we change the value of $R_i$ under different sensing target quantities and obtain three curves. As the number of sensing targets increases, it is seen that the transmit power rises, because more sensing targets mean more beams and greater mutual interference. Besides, we can find that the transmit power shows a subtle increase for small $R_i$ but rises sharply when $R_i$ is close to 1 and the growth becomes greater with more sensing targets. This is due to the fact that MSE is an absolute measure, i.e., for a fixed MSE, estimating a larger true quantity implies better relative accuracy, which in turn requires a higher transmit power. Therefore, it is important to choose the proper required minimum MSE based on the value of $R_i$.

\begin{figure}
    \centering
    \includegraphics[width=1\linewidth]{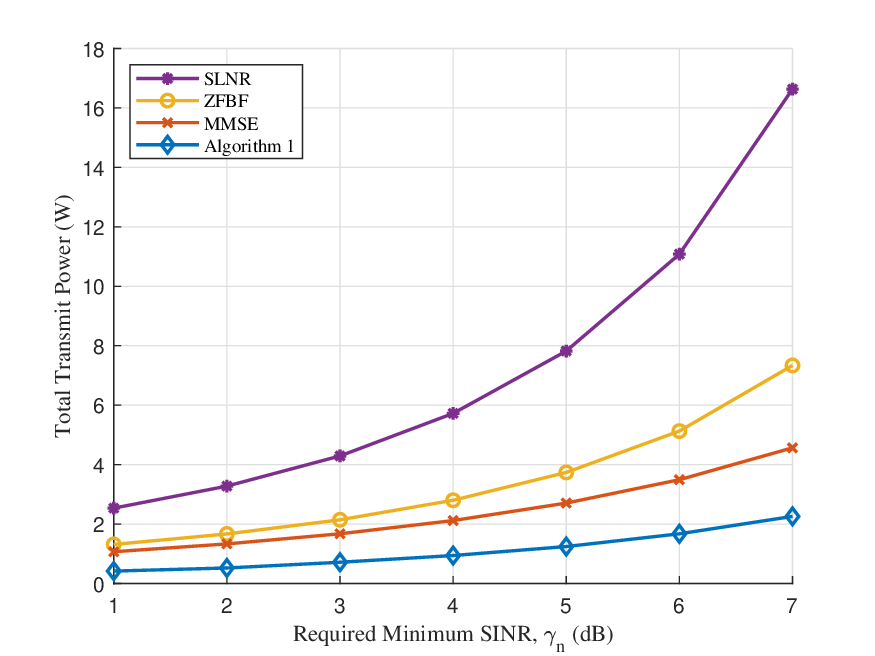}
    \caption{Comparison of different algorithms' performances under different required minimum SINR.}
    \label{fig:diff_algorithm}
\end{figure}

Finally, a performance comparison of different algorithms under different required minimum SINR is presented. We select zero-forcing beamforming (ZFBF) algorithm, MMSE algorithm, and signal-to-leakage-plus-noise ratio (SLNR) algorithm as baseline. All of the three baseline algorithms fix the beams first, and then optimize the total transmit power. Specifically, ZFBF algorithm minimizes the interference, MMSE algorithm reduces the MSE between the desired and received signals \cite{Linear_precoding} and SLNR algorithm maximizes the value of SLNR while obtaining the beamforming vectors \cite{A_leakage-based}. In Fig. \ref{fig:diff_algorithm}, we observe that the Algorithm 1 performs better than baselines, especially with higher SINR, which demonstrates that the proposed algorithm is highly effective. Besides, MMSE algorithm comes next, since the QoS constraints in (\ref{M1}) are sensitive to noise which is considered in MMSE but not in ZFBF. In addition, the SLNR algorithm performs worst because it requires more power to enhance the target signal in order to suppress leakage interference and noise. 

\section{Conclusion}
This paper established an integrated sensing and communication framework for LEO satellite systems. Based on the proposed framework, this paper addressed the key challenge of cross-functional interference exacerbated by imperfect CSI via introducing a robust beamforming design algorithm that enhances both sensing and communication performance in the presence of channel phase uncertainty while considering the limited energy in LEO satellites. Through extensive simulations, we verified the reliability and effectiveness of the proposed algorithm. At the same time, through analysis of several key parameters, we have discovered that it is necessary to balance the proposed algorithm's performance and power consumption when setting the parameters. In practice, the proposed framework can be applied to emerging dual-functional IoT scenarios. For instance, in smart agriculture, the satellite provides broadband downlink to remote farming hubs while simultaneously performing large-scale soil moisture mapping to support precision irrigation. Similarly, for intelligent surveillance applications, it enables high-speed data transmission to monitoring terminals while classifying moving targets to enhance regional security. For future research, extending the single-satellite design to multi-satellite cooperative networks is a promising direction.

\begin{appendices}
\section{Proof of the MSE inequality}
Based on the previous modeling, we can equivalently rewrite equation (\ref{MSE_E}) as
\begin{equation}
\label{MSE_appendix}
\begin{aligned}[b]
\operatorname{MSE}_i^{\text{sens}} 
&= {\left| \mathbf{Y} - 1 \right|}^2 R_i^2 \\
&+ \sum_{\substack{j=1 \\ j \neq i}}^I \sum_{l=1}^I R_j^2 \left| 
        \mathbf{v}_i^{\rm{H}} \operatorname{diag}(\hat{\mathbf{g}}_i) \mathbf{Q}_i^{\prime} \operatorname{diag}(\hat{\mathbf{g}}_i^{\rm{H}}) \mathbf{a}_l 
    \right|^2 \\
&+ \sum_{\substack{j=1 \\ j \neq i}}^I \sum_{n=1}^N R_j^2 \left| 
        \mathbf{v}_i^{\rm{H}} \operatorname{diag}(\hat{\mathbf{g}}_i) \mathbf{Q}_i^{\prime} \operatorname{diag}(\hat{\mathbf{g}}_i^{\rm{H}}) \mathbf{b}_n
    \right|^2 \\
&+ (\sigma^{\prime})^2 \| ({{\bf{v}}_i^{\#}})^{\rm{H}} \|^2, 
\end{aligned}
\end{equation}
where \[
\begin{split}
{\bf{Y}} = {\bf{v}}_i^{\rm{H}} \operatorname{diag}\left( {{\hat{\bf{g}}_i}} \right) {\bf{Q}}_i^{\text{s}} \operatorname{diag}\left( {{\hat{\bf{g}}_i}}^{\rm{H}} \right) 
    \Biggl( & \sum\limits_{l=1}^{I} {{\bf{a}}_l  s_l^{sens}} \\
    & + \sum\limits_{n=1}^N {{\bf{b}}_n s_n^{comm}} \Biggr).
\end{split}
\]

The MSE is made up of four parts, each of which is greater than zero. Since the QoS constraints require us to minimize the MSE as much as possible, each part of MSE should be close to zero. Since the reflection coefficients of most materials are greater than 0.1, we can set the MSE to be less than the mean squared value of the sensing target reflection coefficient. Under these conditions, if $\mathop{\rm Re}\nolimits(\bf{Y})$ is less than 0, MSE will not satisfy the constraints. While $\mathop{\rm Re}\nolimits(\bf{Y})$ is greater than 0, we have 
${\left| {{\bf{Y}} - 1} \right|^2} \le \left|{{\bf{Y}}^2} - 1\right|$. Since other parts of MSE in (\ref{MSE_appendix}) is the same as that in (\ref{MSE_inequality}), we can prove that the inequality in (\ref{MSE_inequality}) is valid.
\end{appendices}

\end{document}